\documentclass[11pt]{article}
\usepackage[utf8]{inputenc}
\usepackage{fullpage}
\usepackage[ruled,vlined]{algorithm2e}
\usepackage{amsmath, amsthm, amssymb, mathtools, color, framed}
\usepackage[noend]{algpseudocode}
\allowdisplaybreaks
\usepackage{hyperref}
\hypersetup{
    colorlinks=true,
    linkcolor=blue,
    filecolor=magenta,      
    urlcolor=cyan,
    citecolor=blue
}

\newtheorem{theorem}{Theorem}[section]
\newtheorem{corollary}[theorem]{Corollary}
\newtheorem{remark}[theorem]{Remark}

\newtheorem{claim}[theorem]{Claim}

\newtheorem{conjecture}[theorem]{Conjecture}
\newtheorem{observation}[theorem]{Observation}
\newtheorem{definition}[theorem]{Definition}

\newcommand{\ket}[1]{|#1\rangle}

\newcommand{\spar}{\mathrm{spar}}
\newcommand{\rank}{\mathrm{rank}}
\newcommand{\bra}[1]{\left\{#1\right\}}
\newcommand{\cbra}[1]{\left\{#1\right\}}
\newcommand{\rbra}[1]{\left(#1\right)}
\newcommand{\mathify}[1]{\ifmmode{#1}\else\mbox{$#1$}\fi}

\newcommand{\textnaandt}{non-adaptive AND decision tree}
\newcommand{\textrnaandt}{randomized non-adaptive AND decision tree}
\newcommand{\textqnaandt}{quantum non-adaptive AND decision tree}
\newcommand{\pmone}{\bra{-1, 1}}
\newcommand{\zone}{\{0, 1\}}

\newcommand{\AND}{\mathsf{AND}}
\newcommand{\XOR}{\mathsf{XOR}}
\newcommand{\sw}{\mathsf{switch}}
\newcommand{\lc}{\lceil}
\newcommand{\rc}{\rceil}

\newcommand{\NAANDT}{\mathsf{NAADT}}
\newcommand{\RNAANDT}{\mathsf{RNAADT}}
\newcommand{\QNAANDT}{\mathsf{QNAADT}}

\newcommand{\NAPDT}{\mathsf{NAPDT}}
\newcommand{\Doneway}{\sD_{\cc}^\rightarrow}
\newcommand{\Dcc}{\sD_{\cc}}

\newcommand{\Ronewaywitherror}[1]{\sR_{\cc, #1}^\rightarrow}
\newcommand{\Qonewaywitherror}[1]{\mathsf{Q}_{\cc, #1}^\rightarrow}
\newcommand{\Qentonewaywitherror}[1]{\mathsf{Q}_{\cc, #1}^{*, \rightarrow}}

\newcommand{\Ddt}{\sD_{\dt}^\rightarrow}
\newcommand{\mpat}{\mathsf{Pat}^{\mathsf{M}}}
\newcommand{\Rdt}{\sR_{\dt}^\rightarrow}
\newcommand{\Qdt}{\sQ_{\dt}^\rightarrow}
\newcommand{\OMB}{\mathsf{OMB}}
\newcommand{\IP}{\mathsf{IP}}

\newcommand{\ADDR}{\mathsf{ADDR}}
\newcommand{\wh}{\widetilde}

\renewcommand{\S}{\mathcal{S}}
\newcommand{\B}{\mathcal{B}}

\newcommand{\X}{\mathcal{X}}

\newcommand{\A}{\mathcal{A}}
\newcommand{\U}{\mathcal{U}}

\newcommand{\Z}{\mathcal{Z}}
\newcommand{\R}{\mathcal{R}}

\newcommand{\Ind}{\mathbb{I}}

\newcommand{\eps}{\varepsilon}
\renewcommand{\epsilon}{\varepsilon}
\newcommand{\sR}{\mathsf{R}}
\newcommand{\sD}{\mathsf{D}}
\newcommand{\sS}{\mathsf{S}}
\newcommand{\sP}{\mathsf{P}}
\newcommand{\sQ}{\mathsf{Q}}
\newcommand{\sI}{\mathcal{I}}
\newcommand{\cc}{\mathrm{cc}}
\newcommand{\VC}{\mathsf{VC}}
\newcommand{\cH}{\mathcal{H}}
\newcommand{\dt}{\mathrm{dt}}
\newcommand{\rH}{\mathbb{H}}
\newcommand{\bin}{\mathsf{bin}}
\newcommand{\agr}{\mathsf{agr}}
\DeclarePairedDelimiter\abs{\lvert}{\rvert}%

\makeatletter
\let\oldabs\abs
\def\abs{\@ifstar{\oldabs}{\oldabs*}}
\let\oldnorm\norm
\def\norm{\@ifstar{\oldnorm}{\oldnorm*}}
\makeatother

\title{One-way communication complexity and non-adaptive decision trees\footnote{This paper includes and improves upon results from an earlier unpublished manuscript of one of the authors~\cite{San17}.}}
\author{
Nikhil S.~Mande\footnote{N.S.M.~is supported by the Dutch Research Council (NWO) through QuantERA project QuantAlgo 680-91-034.}\\
CWI, Amsterdam\\
\textsf{nsm@cwi.nl}
\and
Swagato Sanyal\footnote{S.S.~is supported by an ISIRD Grant from Sponsored Research and Industrial Consultancy, IIT Kharagpur. Part of this work was done while S.S.~was at NTU and CQT, and supported by the Singapore National Research Foundation under NRF RF Award No. NRF-NRFF2013-13.}\\
IIT Kharagpur\\
\textsf{swagato@cse.iitkgp.ac.in}
\and
Suhail Sherif\\
Vector Institute, Toronto\\
\textsf{suhail.sherif@gmail.com}
}

\begin{document}

\maketitle

\begin{abstract}
We study the relationship between various one-way communication complexity measures of a composed function with the analogous decision tree complexity of the outer function. We consider two gadgets: the AND function on 2 inputs, and the Inner Product on a constant number of inputs. More generally, we show the following when the gadget is Inner Product on $2b$ input bits for all $b \geq 2$, denoted $\IP$.
\begin{itemize}
    \item If $f$ is a total Boolean function that depends on all of its $n$ input bits, then the bounded-error one-way quantum communication complexity of $f \circ \IP$ equals $\Omega(n(b-1))$.
    \item If $f$ is a \emph{partial} Boolean function, then the deterministic one-way communication complexity of $f \circ \IP$ is at least $\Omega(b \cdot \Ddt(f))$, where $\Ddt(f)$ denotes non-adaptive decision tree complexity of $f$.
\end{itemize}
To prove our quantum lower bound, we first show a lower bound on the VC-dimension of $f \circ \IP$. We then appeal to a result of Klauck [STOC'00], which immediately yields our quantum lower bound.
Our deterministic lower bound relies on a combinatorial result independently proven by Ahlswede and Khachatrian [Adv.~Appl.~Math.'98], and Frankl and Tokushige [Comb.'99].

It is known due to a result of Montanaro and Osborne [arXiv'09] that the deterministic one-way communication complexity of $f \circ \XOR$ \emph{equals} the non-adaptive parity decision tree complexity of $f$.
In contrast, we show the following when the inner gadget is the AND function on 2 input bits.
\begin{itemize}
    \item There exists a function for which even the \emph{quantum} \textnaandt~complexity of $f$ is exponentially large in the deterministic one-way communication complexity of $f \circ \AND$.
    \item However, for symmetric functions $f$, the \textnaandt~complexity of $f$ is at most quadratic in the (even two-way) communication complexity of $f \circ \AND$.
\end{itemize}
In view of the first bullet, a lower bound on non-adaptive AND decision tree complexity of $f$ \emph{does not} lift to a lower bound on one-way communication complexity of $f \circ \AND$.
The proof of the first bullet above uses the well-studied \emph{Odd-Max-Bit} function. 
For the second bullet, we first observe a connection between the one-way communication complexity of $f$ and the \emph{M\"obius sparsity} of $f$, and then give a lower bound on the M\"obius sparsity of symmetric functions. An upper bound on the \textnaandt~complexity of symmetric functions follows implicitly from prior work on combinatorial group testing; for the sake of completeness, we include a proof of this result.

It is well known that the rank of the communication matrix of a function $F$ is an upper bound on its deterministic one-way communication complexity. This bound is known to be tight for some $F$. However, in our final result we show that this is not the case when $F = f \circ \AND$. More precisely we show that for all $f$, the deterministic one-way communication complexity of $F = f \circ \AND$ is at most $(\rank(M_{F}))(1 - \Omega(1))$, where $M_{F}$ denotes the communication matrix of $F$.
\end{abstract}

\tableofcontents
\newpage

\section{Introduction}\label{sec: intro}

Composed functions are important objects of study in analysis of Boolean functions and computational complexity. For Boolean functions $f:\zone^n \to \zone$ and $g:\zone^m \to \{0,1\}$, their composition $f \circ g:\ \left(\zone^m\right)^n \to \zone$ is defined as follows: $f \circ g(x_1,\ldots, x_n):=f(g(x_1),\ldots, g(x_n))$. In other words, $f \circ g$ is the function obtained by first computing $g$ on $n$ disjoint inputs of $m$ bits each, and then computing $f$ on the $n$ resultant bits. Composed functions have been extensively looked at in the complexity theory literature, with respect to various complexity measures~\cite{BW01, HLS07, Rei11, She12, She13, BT15, Tal13, Mon14, BK16, GJ16, AGJ+17, GLSS19, BB20}.


Of particular interest to us is the case when $g$ is a communication problem (also referred to as ``gadget''). More precisely, let $g : \zone^b \times \zone^b \to \zone$ and $f : \zone^n \to \zone$ be Boolean functions. Consider the following communication problem: Alice has input $x = (x_1, \dots, x_n)$ and Bob has input $y = (y_1 \dots, y_n)$ where $x_i, y_i \in \zone^b$ for all $i \in [n]$. Their goal is to compute $f \circ g((x_1, y_1), \dots, (x_n, y_n))$ using as little communication as possible.
A natural protocol is the following: Alice and Bob jointly simulate an efficient query algorithm for $f$, using an optimal communication protocol for $g$ to answer each query. Lifting theorems are statements that say this naive protocol is essentially optimal. Such theorems enable us to prove lower bounds on the rich model of communication complexity by proving feasibly easier-to-prove lower bounds in the query complexity (decision tree) model.
Various lifting theorems have been proved in the literature~\cite{GLMWZ16, dRNV16, RM99, GPW18, CKLM19, WYY17, GGKS20, GPW20, HHL18, KLMY20, LM19, CFKMP21}.

In this work we are interested in the \emph{one-way} communication complexity of composed functions. Here, a natural protocol is for Alice and Bob to simulate a \emph{non-adaptive} decision tree for the outer function, using an optimal \emph{one-way} communication protocol for the inner function.
Thus, the one-way communication complexity of $f \circ g$ is at most the non-adaptive decision tree complexity of $f$ times the one-way communication complexity of $g$.

Lifting theorems in the one-way model are less studied than in the two-way model. Montanaro and Osborne~\cite{MO09} observed that the deterministic one-way communication complexity of $f \circ \XOR$ \emph{equals} the non-adaptive parity decision tree complexity of $f$. Thus, non-adaptive parity decision tree complexity lifts ``perfectly'' to deterministic communication complexity with the XOR gadget. Kannan et al.~\cite{KMSY18} showed that under uniformly distributed inputs, bounded-error non-adaptive parity decision tree complexity lifts to one-way bounded-error distributional communication complexity with the XOR gadget. Hosseini, Lovett and Yaroslavtsev~\cite{HLY19} showed that randomized non-adaptive parity decision tree complexity lifts to randomized communication complexity with the XOR gadget in the one-way broadcasting model with $\Theta(n)$ players.

We explore the tightness of the naive communication upper bound for two different choices of the gadget $g$: the Inner Product function, and the two-input AND function. For each choice of $g$, we compare the one-way communication complexity of $f \circ g$ with an appropriate type of non-adaptive decision tree complexity of $f$. 
Below, we motivate and state our results for each choice of the gadget. Formal definitions of the measures considered in this section can be found in Section~\ref{sec: prelims}.
\subsection{Inner Product Gadget}
Let $\Qonewaywitherror{\epsilon}(\cdot)$, $\Ronewaywitherror{\epsilon}(\cdot)$ and $\Doneway(\cdot)$ denote quantum $\epsilon$-error, randomized $\epsilon$-error and deterministic one-way communication complexity, respectively. When we allow the parties to share an arbitrary input-independent entangled state in the beginning of the protocol, denote the one-way quantum $\eps$-error communication complexity by $\Qentonewaywitherror{\eps}(\cdot)$. Let $\Qdt(\cdot)$ and $\Ddt(\cdot)$ denote bounded-error quantum non-adaptive decision tree complexity and deterministic non-adaptive decision tree complexity, respectively. For an integer $b > 0$, let $\IP : \zone^b \times \zone^b \to \zone$ denote the \emph{Inner Product Modulo 2} function, that outputs the parity of the bitwise AND of two $b$-bit input strings. Our first result shows that if $f$ is a total function that depends on all of its input bits, the quantum (and hence, randomized) bounded-error one-way communication complexity of $f \circ \IP$ is $\Omega(n(b-1))$. Let $\rH_\bin(\cdot)$ denote the binary entropy function. If $\eps = 1/2 - \Omega(1)$, then $1 - \rH_{\bin}(\eps) = \Omega(1)$. 

\begin{theorem}\label{thm:main}
Let $f:\{0,1\}^n \rightarrow \{0,1\}$ be a total Boolean function that depends on all its inputs (i.e., it is not a junta on a strict subset of its inputs), and let $\epsilon \in (0, 1/2)$. Let $\IP : \zone^b \times \zone^b \to \zone$ denote the Inner Product function on $2b$ input bits for $b \geq 1$. Then $\Qonewaywitherror{\epsilon}(f \circ \IP) \geq (1-\rH_\bin(\epsilon))n(b-1)$ and $\Qentonewaywitherror{\epsilon}(f \circ \IP) \geq (1-\rH_\bin(\epsilon))n(b-1)/2$.
\end{theorem}

\begin{remark}
In an earlier manuscript~\cite{San17}, the second author proved a lower bound of $(1-\rH_\bin(\epsilon))n(b-1)$ on a weaker complexity measure, namely $\Ronewaywitherror{\eps}(F)$, via information-theoretic tools. Kundu~\cite{Kun17} subsequently observed that a quantum lower bound can also be obtained by additionally using Holevo's theorem. They also suggested to the second author via private communication that one might be able to recover these bounds using a result of Klauck~\cite{Kla00}. This is indeed the approach we take, and we thank them for suggesting this and pointing out the reference.
\end{remark}

In order to prove Theorem~\ref{thm:main}, we appeal to a result of Klauck~\cite[Theorem 3]{Kla00}, who showed that the one-way $\epsilon$-error quantum communication complexity of a function $F$ is at least $(1 - \rH_{\bin}(\eps)) \cdot \VC(F)$, where $\VC(F)$ denotes the VC-dimension of $F$ (see Definition~\ref{defn: VC}). In the case when the parties can share an arbitrary entangled state in the beginning of a protocol, Klauck showed a lower bound of $(1 - \rH_{\bin}(\eps)) \cdot \VC(F)/2$. We exhibit a set of inputs that witnesses the fact that $\VC(f \circ \IP) \geq n(b-1)$.
Note that Theorem~\ref{thm:main} is useful only when $b>1$. Indeed, no non-trivial lifting statement is true for $b=1$ when $f$ is the AND function on $n$ bits, since in this case, $f \circ \IP = \AND_{2n}$, whose one-way communication complexity is 1.

Our second result with the Inner Product gadget relates the deterministic one-way communication complexity of $f \circ \IP$ to the deterministic non-adaptive decision tree complexity of $f$, where $f$ is an arbitrary \emph{partial} Boolean function.
\begin{theorem}\label{thm:detlift}
Let $\sS \subseteq \{0,1\}^n$ be arbitrary, and $f: \sS \rightarrow \{0,1\}$ be a \emph{partial} Boolean function. Let $b \geq 2$ and $\IP : \zone^b \times \zone^b \to \zone$. Then $\Doneway(f \circ \IP)=\Omega(b \cdot \Ddt(f))$.
\end{theorem}
Given a protocol $\Pi$, our proof extracts a set of variables of cardinality linear in the complexity of $\Pi$, whose values always determine the value of $f$. The following claim which follows directly from a result due to Ahlswede and Khachatrian~\cite{AK98} and independently Frankl and Tokushige \cite{FT99}, is a crucial ingredient in our proof.

\begin{theorem}\label{thm: packing}
Let $q \geq 3$ and $1 \leq d \leq n/3$. Let $\A \subseteq [q]^n$ be such that for all $x^{(1)}=(x^{(1)}_1, \ldots, x^{(1)}_n)$, $x^{(2)}=(x^{(2)}_1, \ldots, x^{(2)}_n) \in \A$, $|\{i \in [n] \mid x^{(1)}_i = x^{(2)}_i\}| \geq d$. Then, $|\A| < q^{n-\frac{d}{10}}$.
\end{theorem}
We prove Theorem~\ref{thm: packing} in Appendix~\ref{erdos}.
Theorem~\ref{thm: packing} admits simple proofs when $q$ is large compared to $n$. See \cite{GMWW17} for a proof when $q$ is a prime power, and $q \geq n$. Their proof is based on polynomials over finite fields. We give a different proof for all $q > (\frac{en}{d})^2$ in Appendix~\ref{erdos} (for every $\delta >0$, the proof can be extended to work for all $q = \Omega(n/d)^{1+\delta}$). However, such a statement that only holds for large $q$ will only enable us to prove a lifting theorem for a gadget of size $b=\Omega(\log n)$. To prove Theorem~\ref{thm:detlift} for constant-sized gadgets we need to set $q$ to $O(1)$.

\begin{remark}
An analogous lifting theorem for deterministic one-way protocols for \emph{total} outer functions follows as a special case of both Theorem~\ref{thm:main} and Theorem~\ref{thm:detlift}. However, the statement admits a simple and direct proof based on a fooling set argument.
\end{remark}

Theorem~\ref{thm:main} and Theorem~\ref{thm:detlift} give lower bounds even when the gadget is the Inner Product function on 4 input bits (and lower bounds do not hold for the Inner Product gadget with fewer inputs).
It is worth mentioning here that prior works that consider lifting theorems with the Inner Product gadget~\cite{CKLM19, WYY17, CFKMP21}, albeit in the two-way model of communication complexity, require a super-constant gadget size.

\subsection{AND Gadget}
Interactive communication complexity of functions of the form $f \circ \AND$ have gained a recent interest~\cite{KLMY20, Wu21}.
In order to state and motivate our results regarding when the inner gadget is the 2-bit AND function, we first discuss some known results in the case when the inner gadget is the 2-bit XOR function.

Consider non-adaptive decision trees, where the trees are allowed to query arbitrary parities of the input variables. Denote the minimum cost (number of parity queries) of such a tree computing a Boolean function $f$, by $\NAPDT(f)$. An efficient non-adaptive parity decision tree for $f$ can easily be simulated to obtain an efficient deterministic one-way communication protocol for $f \circ \XOR$. Thus, $\Doneway(f \circ \XOR) \leq \NAPDT(f)$. Montanaro and Osborne~\cite{MO09} observed that this inequality is, in fact, tight for all Boolean functions.
More precisely,
\begin{claim}[\cite{MO09}]\label{claim: foxor oneway equals napdt}
For all Boolean functions $f : \zone^n \to \zone$, $\Doneway(f \circ \XOR) = \NAPDT(f)$.
\end{claim}
If the inner gadget were AND instead of XOR, then the natural analogous decision tree model to consider would be non-adaptive decision trees that have query access to arbitrary ANDs of subsets of inputs. Denote the minimum cost (number of AND queries) of such a tree computing a Boolean function $f$ by $\NAANDT(f)$.
Clearly, $\Doneway(f \circ \AND)$ is bounded from above by $\NAANDT(f)$, since a \textnaandt~can be easily simulated to give a one-way communication protocol for $f \circ \AND$ of the same complexity. Thus, $\Doneway(f \circ \AND) \leq \NAANDT(f)$. On the other hand, one can show that $\Doneway(f \circ \AND) \geq \log(\NAANDT(f))$ (see Claim~\ref{claim: doneway vs naandt}). Thus
\begin{equation}\label{eqn: naandt and doneway}
\log(\NAANDT(f)) \leq \Doneway(f \circ \AND) \leq \NAANDT(f).
\end{equation}

We explore if an analogous statement to Claim~\ref{claim: foxor oneway equals napdt} holds true if the inner function were AND instead of XOR. That is, is the second inequality in Equation~\eqref{eqn: naandt and doneway} always tight?

We give a negative answer in a very strong sense and exhibit a function for which the first inequality is tight (up to an additive constant). We show that there is an exponential separation between these measures even if one allows the decision trees to have \emph{quantum} query access to ANDs of subsets of input variables. It is worth noting that, in contrast, if one is given quantum query access to \emph{parities} (in place of ANDs) of subsets of input variables, then one can completely recover an $n$-bit string using just 1 query~\cite{BV97}, rendering this model trivial. Let $\QNAANDT(f)$ denote the bounded-error \textqnaandt~complexity of $f$.
\begin{theorem}\label{thm: intro omb separation}
There exists a function $f : \zone^n \to \zone$ such that $\QNAANDT(f) = \Omega(2^{\Doneway(f \circ \AND)})$.
\end{theorem}
The function $f$ we use to witness the bound in Theorem~\ref{thm: intro omb separation} is a modification of the well-studied \emph{Odd-Max-Bit} function, which we denote $\OMB_n$. This function outputs 1 if and only if the maximum index of the input string that contains a 0, is odd (see Definition~\ref{defn: omb}).
A $\lceil\log (n+1)\rceil$-cost one-way communication protocol is easy to show, since Alice can simply send Bob the maximum index where her input is 0 (if it exists), and Bob can use this along with his input to conclude the parity of the maximum index where the bitwise AND of their inputs is 0.
A crucial property that we use to show a lower bound of $\Omega(n)$ on $\QNAANDT(\OMB_n)$ is that $\OMB_n$ has large \emph{alternating number}, that is, there is a monotone path on the Boolean hypercube from $0^n$ to $1^n$ on which the value of $\OMB_n$ flips many times.

Theorem~\ref{thm: intro omb separation} implies that, in contrast to the lifting theorem with the XOR gadget (Claim~\ref{claim: foxor oneway equals napdt}), the measure of \textnaandt~complexity \emph{does not} lift to a one-way communication lower bound for $f \circ \AND$.
However we show that a statement analogous to Claim~\ref{claim: foxor oneway equals napdt} does hold true for symmetric functions $f$, albeit with a quadratic factor, even when the measure is two-way communication complexity, denoted $\Dcc(\cdot)$.

\begin{theorem}\label{thm: intro symmetric naandt communication}
Let $f : \zone^n \to \zone$ be a symmetric function. Then $\NAANDT(f) = O(\Dcc(f \circ \AND)^2)$.
\end{theorem}
In fact we prove a stronger bound in which $\Dcc(f \circ \AND)$ above is replaced by $\log\rank(M_{f \circ \AND})$, where $M_{f \circ \AND}$ denotes the communication matrix of $f \circ \AND$. That is, we show that for symmetric functions $f$,
\begin{equation}\label{eqn: naandt logrank symmetric}
 \NAANDT(f) = O(\log^2\rank(M_{f \circ \AND})).
\end{equation}
Since it is well known (Equation~\eqref{eqn: logrank bounds}) that the communication complexity of a function is at least as large as the logarithm of the rank of its communication matrix, this implies Theorem~\ref{thm: intro symmetric naandt communication}.
There have been multiple works (see, for example,~\cite{BW01, Wu21, KLMY20} and the references therein) studying the communication complexity of $\AND$ functions in connection with the log-rank conjecture~\cite{LS88} which states that the communication complexity is bounded from above by a polynomial in the logarithm of the rank of the communication matrix.
Among other things, Buhrman and de Wolf~\cite{BW01} observed that the log-rank conjecture holds for symmetric functions composed with AND. In particular, they showed that if $f$ is symmetric, then $\Dcc(f \circ \AND) = O(\log \rank(M_{f \circ \AND})))$. Most recently, Knop et al.~\cite{KLMY20} showed that $\Dcc(f \circ \AND) = O(\textnormal{poly}(\log \rank(M_{f \circ \AND}), \log n)$ for all Boolean functions $f : \zone^n \to \zone$, nearly resolving the log-rank conjecture for $\AND$ functions.

While we have a quadratically worse dependence in the RHS of Equation~\eqref{eqn: naandt logrank symmetric} as compared to the above-mentioned bound for symmetric functions due to Buhrman and de Wolf, our upper bound is on a complexity measure that can be exponentially larger than communication complexity in general (Theorem~\ref{thm: intro omb separation}).

Buhrman and de Wolf showed a lower bound on $\log\rank(M_{f \circ \AND})$ for symmetric functions $f$. An upper bound on $\NAANDT(f)$ implicitly follows from prior work on group testing~\cite{DR83}, but we provide a self-contained probabilistic proof for completeness. Combining these two results yields Equation~\eqref{eqn: naandt logrank symmetric}, and hence Theorem~\ref{thm: intro symmetric naandt communication}.

Suitable analogues of Theorem~\ref{thm: intro omb separation} and Theorem~\ref{thm: intro symmetric naandt communication} can be easily seen to hold when the inner gadget is OR instead of AND. In this case, the relevant decision tree model is non-adaptive OR decision trees. Interestingly, these decision trees are studied in the seemingly different context of \emph{non-adaptive group testing algorithms}.
Non-adaptive group testing is an active area of research (see, for, example,~\cite{CH08} and the references therein), and has additionally gained significant interest of late in view of the ongoing pandemic (see, for example,~\cite{ZLG21}).

Our final result regarding the AND gadget deals with the relationship between one-way communication complexity and rank of the underlying communication matrix.
It is easy to show that for functions $F : \zone^m \times \zone^n \to \zone$,
\begin{equation}\label{eqn: rank bound 1way}
\log\rank(M_F) \leq \Doneway(F) \leq \rank(M_F),
\end{equation}
where $M_F$ denotes the communication matrix of $F$ and is defined by $M_F(x, y) = F(x, y)$, and $\rank(\cdot)$ denotes real rank. The first bound can be seen to be tight for functions with maximal rank, for example the Equality function.
The second inequality is tight, for example, for the Addressing function on $(\log n + n)$ input bits (see Definition~\ref{defn: addressing}) where Alice receives $n$ target bits and Bob receives $\log n$ addressing bits.
Sanyal~\cite{San19} showed that the upper bound can be improved for functions of the form $F = f \circ \XOR$. More precisely they showed that for all Boolean functions $f : \zone^n \to \zone$,
\begin{equation}
\Doneway(f \circ \XOR) \leq O\left(\sqrt{\rank(M_{f \circ \XOR})} \log \rank(M_{f \circ \XOR})\right),
\end{equation}
and moreover this bound is tight up to the logarithmic factor on the RHS, when $f$ is the Addressing function.
We show that the same bound does not hold when the XOR gadget is replaced by AND. We show in Corollary~\ref{cor: doneway vs rank addressing} that when $f$ is the Addressing function, then 
\begin{equation}
\Doneway(f \circ \AND) \geq \rank(M_{f \circ \AND})^{\log_32} \approx \rank(M_{f \circ \AND})^{0.63},
\end{equation}
Thus it is plausible that the upper bound in terms of rank from Equation~\eqref{eqn: rank bound 1way} might be tight for some function of the form $f \circ \AND$. We show that this is not the case.
\begin{theorem}\label{thm: 1way sublinear rank}
Let $f: \zone^n : \zone$ be a Boolean function. Then,
\[
\Doneway(f \circ \AND) \leq (\rank(M_{f \circ \AND}))(1 - \Omega(1)).
\]
\end{theorem}
We show that $\Doneway(f \circ \AND)$ is equal to the logarithm of a measure that we define in this work: the \emph{M\"obius pattern complexity} of $f$, which is the total number of distinct evaluations of the monomials in the M\"obius expansion of $f$ (see Section~\ref{subsec: expansions} for a formal definition of M\"obius expansion).
\begin{definition}[M\"obius pattern complexity]\label{defn: patcomp}
Let $f : \zone^n \to \zone$ be a Boolean function, and let $f = \sum_{S \in \S_f}\wh{f}(S)\AND_S$ be its M\"obius expansion. For an input $x \in \zone^n$, define the \emph{pattern of $x$} to be $\rbra{\AND_S(x)}_{S \in \S_f}$. Define the \emph{M\"obius pattern complexity} of $f$, denoted $\mpat(f)$, by $\mpat(f) := \abs{\cbra{P \in \zone^{\S_f} : P = \rbra{\AND_S(x)}_{S \in \S_f}~\textnormal{for some}~x \in \zone^n}}$.
\end{definition}
When clear from context, we refer to the M\"obius pattern complexity of $f$ just as the pattern complexity of $f$.

All of our results involving bounds for $\Doneway(f \circ \AND)$ use the above-mentioned equivalence between it and $\log(\mpat(f))$ (see Claim~\ref{claim: doneway equals log pattern}). We unravel interesting mathematical structure in the M\"obius supports of Boolean functions, and use them to bound their pattern complexity. We hope that pattern complexity will prove useful in future research.


\subsection{Organization}
We introduce the necessary preliminaries in Section~\ref{sec: prelims}. In Section~\ref{sec: IP composition} we prove our results regarding the Inner Product gadget (Theorem~\ref{thm:main} and Theorem~\ref{thm:detlift}). In Section~\ref{sec: naandt is not communication} we prove our results regarding the AND gadget (Theorem~\ref{thm: intro omb separation} and Theorem~\ref{thm: intro symmetric naandt communication}). In Appendix~\ref{app: addressing} we show some results regarding the Addressing function. We provide a combinatorial proof missing from the main text in Appendix~\ref{erdos}.


\section{Preliminaries}\label{sec: prelims}

All logarithms in this paper are taken base 2. We use the notation $[n]$ to denote the set $\{1, \ldots, n\}$. We often identify subsets of $[n]$ with their corresponding characteristic vectors in $\zone^n$. The view we take will be clear from context. Let $\sS \subseteq \{0,1\}^n$ be an arbitrary subset of the Boolean hypercube, and let $f:\sS \rightarrow \{0,1\}$ be a partial Boolean function. If $\sS=\{0,1\}^n$, then $f$ is said to be a total Boolean function. When not explicitly mentioned otherwise, we assume Boolean functions to be total.

\begin{definition}[Binary entropy]
For $p \in (0,1)$, the binary entropy of $p$, $\rH_\bin(p)$, is defined to be the Shannon entropy of a random variable taking two distinct values with probabilities $p$ and $1-p$.
\[\rH_\bin(p):=p \log \frac{1}{p} + (1-p) \log \frac{1}{1-p}.\]
\end{definition}

We define the Inner Product Modulo 2 function on $2b$ input bits, denoted $\IP$ (we drop the dependence of $\IP$ on $b$ for convenience; the value of $b$ will be clear from context).

\begin{definition}[Inner Product Modulo 2]\label{defn: ip}
For an integer $b > 0$, define the \emph{Inner Product Modulo 2} function, denoted $\IP : \zone^b \times \zone^b \to \zone$ by $\IP(x_1, \dots, x_b, y_1, \dots, y_b) = \oplus_{i \in [b]}(\AND(x_i, y_i))$.
\end{definition}

If $f$ is a partial function, so is $f \circ \IP$. We also define the following functions.

\begin{definition}[Odd-Max-Bit]\label{defn: omb}
Define the \emph{Odd-Max-Bit} function,\footnote{In the literature, $\OMB_n$ is typically defined with a 1 in the max instead of 0. That function behaves very differently from our $\OMB_n$. For example, it is known that even the \emph{weakly unbounded-error} communication complexity of $\OMB_n \circ \AND$ (under the standard definition of $\OMB_n$) is polynomially large in $n$~\cite{BVW07}. In contrast, it is easy to show that even the deterministic one-way communication complexity of $\OMB_n \circ \AND$ equals $\lc\log(n+1)\rc$ with our definition (see Theorem~\ref{thm: omb and exponential separation}).} denoted $\OMB_n : \zone^n \to \zone$, by $\OMB_n(x) = 1$ if $\max\bra{i \in [n] : x_i=0}$ is odd, and $\OMB_n(x) = 0$ otherwise. Define $\OMB_n(1^n) = 0$.
\end{definition}

\begin{definition}[Addressing]\label{defn: addressing}
For an integer $n \geq 2$ that is a power of 2, define the \emph{Addressing} function, denoted $\ADDR_n : \zone^{\log n + n} \to \zone$, by
\[
\ADDR_n(x, y) = y_{\bin(x)},
\]
where $\bin(x)$ denotes the integer in $[n]$ whose binary representation is $x$. We refer to the $x$-variables as \emph{addressing variables} and the $y$-variables as \emph{target variables}.
\end{definition}

\subsection{Decision Trees and Their Variants}
For a partial Boolean function $f:\sS \rightarrow \zone$, the deterministic non-adaptive query complexity (alternatively the non-adaptive decision tree complexity) $\Ddt(f)$ is the minimum integer $k$ such that the following is true: there exist $k$ indices $i_1, \ldots, i_k \in [n]$, such that for every Boolean assignment $a_{i_1}, \ldots, a_{i_k}$ to the input variables $x_{i_1}, \ldots, x_{i_k}$, $f$ is constant on $\sS \cap \{x \in \{0,1\}^n \mid \forall j=1, \ldots, k, x_{i_j}=a_{i_j}\}$. Equivalently $\Ddt(f)$ is the minimum number of variables such that $f$ can be expressed as a function of these variables. It is easy to see that if $f$ is a total function that depends on all input variables, then $\Ddt(f)=n$.

\begin{definition}[Non-adaptive parity decision tree complexity]\label{defn: napdt}
Define the \emph{non-adaptive parity decision tree complexity} of $f : \zone^n \to \zone$, denoted by $\NAPDT(f)$, to be the minimum number of parities such that $f$ can be expressed as a function of these parities.
In other words, the non-adaptive parity decision tree complexity of $f$ equals the minimal number $k$ for which there exists $\S = \bra{\bra{S_1, \dots, S_k} : S_i \subseteq [n]~\text{for all}~i \in [k]}$ such that the function value $f(x)$ is determined by the values $\bra{\oplus_{j \in S_{i}}x_j : i \in [k]}$ for all $x \in \zone^n$.
\end{definition}

\begin{definition}[Non-adaptive AND decision tree complexity]\label{defn: naandt}
Define the \emph{\textnaandt~complexity} of $f : \zone^n \to \zone$, denoted by $\NAANDT(f)$, to be the minimum number of monomials such that $f$ can be expressed as a function of these monomials.
In other words, the \textnaandt~complexity of $f$ equals the minimal number $k$ for which there exists $\S = \bra{\bra{S_1, \dots, S_k} : S_i \subseteq [n]~\text{for all}~i \in [k]}$ such that the function value $f(x)$ is determined by the values $\bra{\AND_{S_{i}}(x) : i \in [k]}$ for all $x \in \zone^n$. We refer to such a set $\S$ as an \emph{NAADT basis} for $f$.
\end{definition}

\begin{definition}[Randomized non-adaptive AND decision tree complexity]\label{defn: rnaandts}
A \textrnaandt~$T$ computing $f$ is a distribution over \textnaandt s with the property that $\Pr[T(x) = f(x) \geq 2/3]$ for all $x \in \zone^n$. The cost of $T$ is the maximum cost of a \textnaandt~in its support.
Define the \emph{\textrnaandt~complexity} of $f : \zone^n \to \zone$, denoted by $\RNAANDT(f)$, to be the minimum cost of a \textrnaandt~that computes $f$.
\end{definition}
We refer the reader to~\cite{NC00} for the basics of quantum computing.
\begin{definition}[Quantum non-adaptive AND decision tree complexity]\label{def: qnaandt}
A \textqnaandt~of cost $c$ is a query algorithm that works with a state space $\ket{S_1, \dots, S_c}\ket{b}\ket{w}$, where each $S_j \subseteq [n]$, $b \in \zone^c$ and the last register captures a workspace of an arbitrary dimension. It is specified by a starting state $\ket{\psi}$ and a projective measurement $\cbra{\Pi,I-\Pi}$. For an input $x \in \zone^n$, the action of the non-adaptive query oracle $O_x^{\otimes c}$ is captured by its action on the basis states, described below.
\[
O_x^{\otimes c} \ket{S_1, \dots, S_c}\ket{b_1, \dots, b_c}\ket{w} \mapsto \ket{S_1, \dots, S_c}\ket{b_1 \oplus \AND_{S_1}(x), \dots, b_c \oplus \AND_{S_c}(x)}\ket{w}.
\]
We use $O_x$ to refer to this oracle since $c$ is already unambiguously determined by the state space.
The algorithm accepts $x$ with probability $\| \Pi O_x \ket{\psi} \|^2$.

Define the \emph{\textqnaandt~complexity} of $f : \zone^n \to \zone$, denoted by $\QNAANDT(f)$, to be the minimum cost of a \textqnaandt~that outputs the correct value of $f(x)$ with probability at least $2/3$ for all $x \in \zone^n$.
\end{definition}

The quantum non-adaptive query complexity, denoted $\Qdt$, is defined similarly, the only difference being that the sets $S_1, \dots, S_c$ are restricted to be singletons. Montanaro~\cite{Mon10} observed that $\Qdt(f) = \Omega(n)$ for all total Boolean functions $f : \zone^n \to \zone$ that depend on all input bits. Our proofs of Theorem~\ref{thm: quantum NAANDT OMB} and Claim~\ref{claim: addr qnaandt large} use ideas from their proof.

We require the following well-known fact about quantum algorithms. \begin{claim}\label{claim: final states far}
Let $\mathcal{A}$ be a two-outcome quantum algorithm computing a (possibly partial) Boolean function $f : \sS \to \zone$ to error $\eps$, and let $\cbra{\Pi, I - \Pi}$ be the final measurement it performs. Let $\ket{\psi_z}$ denote the final state of the algorithm on input $z \in \sS$ before measurement. Then for any $x \in f^{-1}(1)$ and $y \in f^{-1}(0)$,
\[
\|\ket{\psi_x} - \ket{\psi_y}\|^2 \geq 2 - 4\sqrt{\eps(1 - \eps)}.
\]
\end{claim}
We provide a proof below for completeness.
\begin{proof}
The correctness of the algorithm guarantees that $\|\Pi\ket{\psi_x}\|^2 \geq 1 - \eps$ and $\|\Pi\ket{\psi_y}\|^2 \leq \eps$. Thus we have
\begin{align*}
    \|\ket{\psi_x} - \ket{\psi_y}\|^2 & = \|\Pi(\ket{\psi_x} - \ket{\psi_y})\|^2 + \|(I - \Pi)(\ket{\psi_x} - \ket{\psi_y})\|^2 \tag*{since $\Pi$ and $I - \Pi$ are orthogonal projectors}\\
    & \geq \rbra{\|\Pi\ket{\psi_x}\| - \|\Pi\ket{\psi_y}\|}^2 + \rbra{\|(I - \Pi)\ket{\psi_x}\| - \|(I - \Pi)\ket{\psi_y}\|}^2 \tag*{by the triangle inequality}\\
    & \geq 2(\sqrt{1 - \eps} - \sqrt{\eps})^2 \tag*{by assumption}\\
    & = 2 - 4\sqrt{\eps(1 - \eps)}.
\end{align*}
\end{proof}

\subsection{M\"obius and Fourier Expansions of Boolean Functions}\label{subsec: expansions}
Every Boolean function $f : \zone^n \to \zone$ has a unique expansion as $f = \sum_{S \subseteq [n]}\wh{f}(S) \AND_S$, where $\AND_S$ denotes the AND of the input variables in $S$ and each $\wh{f}(S)$ is a real number.
We refer to the functions $\AND_S$ as \emph{monomials}, the expansion as the \emph{M\"obius expansion} of $f$, and the real coefficients $\wh{f}(S)$ as the M\"obius coefficients of $f$.
It is known~\cite{Bei93} that the M\"obius coefficients can be expressed as
$\wh{f}(S) = \sum_{X \subseteq S}(-1)^{|S \setminus X|}f(X)$. Define the \emph{M\"obius support} of $f$, denoted $\S_f$, to be the set $\S_f := \bra{S \subseteq [n] : \wh{f}(S) \neq 0}$.
Define the \emph{M\"obius sparsity} of $f$, denoted $\spar(f)$, to be $\spar(f) := \left\lvert \S_f\right\rvert$.

Every Boolean function $f: \zone^n \rightarrow \mathbb{R}$ can also be uniquely written as $f= \sum\limits_{S \subseteq [n]}\widehat{f}(S)(-1)^{\oplus_{j \in S}x_j}$. This representation is called the \emph{Fourier expansion} of $f$ and the real values $\widehat{f}(S)$ are called the Fourier coefficients of $f$. The Fourier sparsity of $f$ is defined to be number of non-zero Fourier coefficients of $f$.
Sanyal~\cite{San19} showed the following relationship between non-adaptive parity decision complexity of a Boolean function and its Fourier sparsity.
\begin{theorem}[{\cite{San19}}]\label{thm: san19}
Let $f : \zone^n \to \pmone$ be a Boolean function with Fourier sparsity $r$. Then $\NAPDT(f) = O(\sqrt{r} \log r)$.
\end{theorem}
This theorem is tight up to the logarithmic factor, witnessed by the Addressing function. We now note some simple observations relating the \textnaandt~complexity of Boolean functions and their M\"obius expansions.

\begin{claim}\label{claim: naandt basis monomial product}
Let $f : \zone^n \to \zone$ be a Boolean function and let $\S = \bra{S_1, \dots, S_k}$ be a NAADT basis for $f$. Then, every monomial in the M\"obius support of $f$ equals $\prod_{i \in T}\AND_{S_i}$, for some $T \subseteq [k]$. 
\end{claim}

\begin{proof}
Since $\S$ is an NAADT basis for $f$, the values of $\bra{\AND_{S_i} : i \in [k]}$ determine the value of $f$. That is, we can express $f$ as
\[
f = \sum_{T \subseteq [k]}b_T \prod_{i \in T}\AND_{S_i}\prod_{j \notin T}(1 - \AND_{S_j}),
\]
for some values of $b_T \in \zone$. 
Expanding this expression only yields monomials that are products of $\AND_{S_i}$'s from $\S$. The claim now follows since the M\"obius expansion of a Boolean function is unique.
\end{proof}

\begin{claim}\label{claim: trivial bounds}
Let $f : \zone^n \to \zone$ be a Boolean function with M\"obius sparsity $r$. Then $\log r \leq \NAANDT(f) \leq r$.
\end{claim}

\begin{proof}
The upper bound $\NAANDT(f) \leq r$ follows from the fact that knowing the values of all ANDs in the M\"obius support of $f$ immediately yields the value of $f$ by plugging these values in the M\"obius expansion of $f$. That is, the M\"obius support of $f$ acts as an NAADT basis for $f$.

For the lower bound, let $\NAANDT(f) = k$, and let $\S = \bra{S_1, \dots, S_k}$ be an NAADT basis for $f$. Claim~\ref{claim: naandt basis monomial product} implies that every monomial in the M\"obius expansion of $f$ is a product of some of these $\AND_{S_i}$'s. Thus, the M\"obius sparsity of $f$ is at most $2^k$, yielding the required lower bound.
\end{proof}

\subsection{Communication Complexity}
The standard model of two-party communication complexity was introduced by Yao~\cite{Yao79}. In this model, there are two parties, say Alice and Bob, each with inputs $x, y \in \zone^n$. They wish to jointly compute a function $F(x, y)$ of their inputs for some function $F : \U \to \zone$ that is known to them, where $\U$ is a subset of $\zone^n \times \zone^n$. They use a communication protocol agreed upon in advance. The cost of the protocol is the number of bits exchanged in the worst case (over all inputs). Alice and Bob are required to output the correct answer for all inputs $(x, y) \in \U$.
The communication complexity of $F$ is the best cost of a protocol that computes $F$, and we denote it by $\Dcc(F)$. See, for example,~\cite{KN97}, for an introduction to communication complexity.

In a deterministic one-way communication protocol, Alice sends a message $m(x)$ to Bob. Then Bob outputs a bit depending on $m(x)$ and $y$.  The complexity of the protocol is the maximum number of bits a message contains for any input $x$ to Alice. In a randomized one-way protocol, the parties share some common random bits $\R$. Alice's message is a function of $x$ and $\R$. Bob's output is a function of $m(x), y$ and $\R$. The protocol $\Pi$ is said to compute $F$ with error $\epsilon \in (0, 1/2)$ if for every $(x,y) \in \U$, the probability over $\R$ of the event that Bob's output equals $F(x,y)$ is at least $1-\epsilon$. The cost of the protocol is the maximum number of bits contained in Alice's message for any $x$ and $\R$.
In the one-way quantum model, Alice sends Bob a quantum message, after which Bob performs a projective measurement and outputs the measurement outcome. Depending on the model of interest, Alice and Bob may or may not share an arbitrary input-independent entangled state for free. We refer the reader to~\cite{Wol02} for an introduction to quantum communication complexity. As in the randomized setting, a protocol $\Pi$ computes $F$ with error $\eps$ if $\Pr[\Pi(x, y) \neq f(x, y)] \leq \eps$ for all $(x, y) \in \U$.

The deterministic ($\eps$-error randomized, $\eps$-error quantum, $\eps$-error quantum with entanglement, respectively) one-way communication complexity of $F$, denoted by $\Doneway(\cdot)$ ($\Ronewaywitherror{\epsilon}(\cdot)$, $\Qonewaywitherror{\eps}(\cdot)$, $\Qentonewaywitherror{\eps}(\cdot)$, respectively), is the minimum cost of any deterministic ($\epsilon$-error randomized, $\eps$-error quantum, $\eps$-error quantum with entanglement, respectively) one-way communication protocol for $F$.

Total functions $F$ whose domain is $\zone^n \times \zone^n$ induce a communication matrix $M_F$ whose rows and columns are indexed by strings in $\zone^n$, and the $(x, y)$'th entry equals $F(x, y)$. It is known that
\begin{equation}\label{eqn: logrank bounds}
\log\rank(M_F) \leq \Dcc(F) \leq O(\sqrt{\rank(M_F)} \log\rank(M_F)),
\end{equation}
where $\rank(\cdot)$ denotes real rank. The first inequality is well known (see, for instance~\cite{KN97}), and the second inequality was shown by Lovett~\cite{Lov16}. One of the most famous conjectures in communication complexity is the log-rank conjecture, due to Lov{\'{a}}sz and Saks~\cite{LS88}, that proposes that the communication complexity of any Boolean function is polylogarithmic in its rank, i.e.~the first inequality in Equation~\eqref{eqn: logrank bounds} is always tight up to a polynomial dependence.

Buhrman and de Wolf~\cite{BW01} observed that the M\"obius sparsity of a Boolean function $f$ equals the rank of the communication matrix of $f \circ \AND$. That is, for all Boolean functions $f: \zone^n \to \zone$,
\begin{equation}\label{eqn: mobius sparsity rank}
    \spar(f) = \rank(M_{f \circ \AND}).
\end{equation}
In view of the first inequality in Equation~\eqref{eqn: logrank bounds}, this yields
\begin{equation}\label{eqn: log mobius sparsity communication lower bound}
    \Dcc(f \circ \AND) \geq \log(\spar(f)).
\end{equation}

We require the definition of the Vapnik-Chervonenkis (VC) dimension~\cite{VC}.
\begin{definition}[VC-dimension]\label{defn: VC}
Consider a function $F : \zone^n \times \zone^n \to \zone$. A subset of columns $C$ of $M_F$ is said to be \emph{shattered} if all of the $2^{|C|}$ patterns of 0's and 1's are attained by some row of $M_F$ when restricted to the columns $C$.
The \emph{VC-dimension} of a function $F : \zone^n \times \zone^n$, denoted $\VC(F)$, is the maximum size of a shattered subset of columns of $M_F$.
\end{definition}

Klauck~\cite{Kla00} showed that the one-way quantum communication complexity of a function $F$ is bounded below by the VC-dimension of $F$.
\begin{theorem}[{\cite[Theorem 3]{Kla00}}]\label{thm: klauck}
Let $F : \zone^n \times \zone^n \to \zone$ be a Boolean function. Then,
$\Qonewaywitherror{\epsilon}(F) \geq (1-\rH_\bin(\epsilon))\VC(F)$ and $\Qentonewaywitherror{\epsilon}(F) \geq (1-\rH_\bin(\epsilon))\VC(F)/2$.
\end{theorem}

\section{Composition with Inner Product}\label{sec: IP composition}
In this section we prove Theorem~\ref{thm:main} and Theorem~\ref{thm:detlift}, which are our results regarding the quantum and deterministic one-way communication complexities, respectively, of functions composed with a small Inner Product gadget.
\subsection{Quantum Complexity}
\begin{proof}[Proof of Theorem~\ref{thm:main}]
By Theorem~\ref{thm: klauck}, it suffices to show that $\VC(f \circ \IP) \geq n(b-1)$. Since $f$ is a function that depends on all its input variables, the following holds. For each index $i \in [n]$, there exist inputs $z^{(i, 0)} = z^{(i)}_1, \dots, z^{(i)}_{i-1}, 0, z^{(i)}_{i+1}, \dots, z^{(i)}_{n}$ and $z^{(i, 1)} = z^{(i)}_1, \dots, z^{(i)}_{i-1}, 1, z^{(i)}_{i+1}, \dots, z^{(i)}_{n}$ such that $f(z^{(i, 0)}) = v_i$ and $f(z^{(i, 1)}) = 1 - v_i$. That is, $z^{(i, 0)}$ and $z^{(i, 1)}$ have different function values, but differ only on the $i$'th bit.

For each $i \in [n]$ and $j \in \bra{2, 3, \dots, b}$, define a string $y^{(i, j)} \in \zone^{nb}$ as follows. For all $k \in [n]$ and $\ell \in [b]$,
\[
y^{(i, j)}_{k, \ell} = \begin{cases}
z^{(i)}_k & \text{if}~k \neq i~\text{and}~\ell = 1\\
1 & \text{if}~k = i~\text{and}~\ell = j\\
0 & \text{otherwise}.
\end{cases}
\]
That is, for $k \neq i$, the $k$'th block of $y^{(i, j)}$ is $(z^{(i)}_k, 0^{b-1})$, and the $i$'th block of $y^{(i, j)}$ is $(0^{j-1}, 1, 0^{b - j})$.
Consider the set of $n(b-1)$-many columns of $M_{f \circ \IP}$, one for each $y^{(i, j)}$. We now show that this set of columns is shattered.
Consider an arbitrary string $c = c_{1, 2}, \dots, c_{1, b}, \dots, c_{n, 2}, \dots, c_{n, b} \in \zone^{n(b-1)}$.
We now show the existence of a row that yields this string on restriction to the columns described above. Define a string $x \in \zone^{nb}$ as follows. For all $i \in [n]$ and $j \in [b]$, $x_{i, 1} = 1$ and
\begin{align*}
x_{i, j} & = \begin{cases}
c_{i, j} & \text{if}~v_i = 0\\
1 - c_{i, j} & \text{if}~v_i = 1.
\end{cases}
\end{align*}
That is, the first element of each block of $x$ is $1$, and the remaining part of any block, say the $i$'th block, equals either the string $c_{i, 2}, \dots, c_{i, b}$ or its bitwise negation, depending on the value of $v_i$.

To complete the proof, we claim that the row of $M_{f \circ \IP}$ corresponding to this string $x$ equals the string $c$ when restricted to the columns $\bra{y^{(i, j)}}_{i \in [n], j \in \bra{2, 3, \dots, b}}$.
To see this, fix $i \in [n]$ and $j \in \bra{2, 3, \dots, b}$ and consider $M_{f \circ \IP}(x, y^{(i, j)})$. Next, for each $k \in [n]$ with $k \neq i$, the inner product of the $k$'th block of $x$ with the $k$'th block of $y$ equals $z^{(i)}_k$, since $x_{k, 1} = 1$ and the first element of the $k$'th block of $y^{(i, j)}$ equals $z^{(i)}_{k}$, and all other elements of the block are 0 by definition. In the $i$'th block of $y^{(i, j)}$, only the $j$'th element is non-zero, and equals 1 by definition. Moreover, $x_{i, j} = c_{i, j}$ if $v_i = 0$, and equals $1 - c_{i, j}$ otherwise. Hence, the inner products of the $i$'th blocks of $x$ and $y^{(i, j)}$ equals $c_{i, j}$ if $v_i = 0$, and equals $1 - c_{i, j}$ otherwise. Thus, the string obtained on taking the block-wise inner product of $x$ and $y^{(i, j)}$ equals $z^{(i)}_1, \dots, z^{(i)}_{i - 1}, c_{i, j}, z^{(i)}_{i + 1}, \dots, z^{(i)}_{n}$ if $v_i = 0$ and $z^{(i)}_1, \dots, z^{(i)}_{i - 1}, 1 - c_{i, j}, z^{(i)}_{i + 1}, \dots, z^{(i)}_{n}$ if $v_i = 1$.
By our definitions of $z^{(i, 0)}, z^{(i, 1)}$ and $v_i$ for each $i \in [n]$, it follows that the value of $f$ when applied to either of these inputs equals $c_{i, j}$. This concludes the proof.
\end{proof}

\subsection{Deterministic Complexity}
We now prove Theorem~\ref{thm:detlift}, which gives a lower bound on the deterministic one-way communication complexity of $f \circ \IP$ for partial functions $f$.
A crucial ingredient of our proof is Theorem~\ref{thm: packing}, which we prove in Appendix~\ref{erdos}. Now we proceed to the proof of Theorem~\ref{thm:detlift}.
\begin{proof}[Proof of Theorem~\ref{thm:detlift}]
Let $q:=2^b-1$ and let $\Pi$ be an optimal one-way deterministic protocol for $f \circ \IP$ of complexity $\Doneway(f \circ \IP)=:c\log q$. The theorem is trivially true if $c \geq n/30$ since $\Ddt(f) \leq n$. In the remainder of the proof we assume that $c < n/30$. $\Pi$ induces a partition of $\{0,1\}^{nb}$ into at most $q^c$ parts; each part corresponds to a distinct message. There are $(2^b-1)^n = q^n$ inputs $(x_1, \ldots, x_n)$ to Alice such that for each $i$, $x_i \neq 0^b$. Let $\Z$ be the set of those inputs. Identify $\Z$ with $[q]^n$. 
By the pigeon-hole principle there exists one part $\sP$ in the partition induced by $\Pi$ that contains at least $q^{n-c}$ strings in $\Z$. We now invoke Theorem~\ref{thm: packing} with $d$ set to $10c$. This is applicable since $d \leq n/3$ and the assumption $b \geq 2$ implies that $q \geq 3$. Theorem~\ref{thm: packing} implies that there are two strings $x^{(1)} = (x^{(1)}_1, \dots, x^{(1)}_n), x^{(2)}=(x^{(2)}_1, \ldots, x^{(2)}_n) \in \sP \cap \Z$ such that $|\{i \in [n] \mid x^{(1)}_i = x^{(2)}_i\}| < 10c$. Let $\sI:=\{i \in [n] \mid x^{(1)}_i = x^{(2)}_i\}$. 
Let $z=(z_1, \dots, z_n)$ denote a generic input to $f$. We claim that for each Boolean assignment $(a_{i})_{i \in \sI}$ to the variables in $\sI$, $f$ is constant on $\sS \cap \{z: \forall i \in \sI, z_i=a_i\}$. 
This will prove the theorem, since querying the variables $\{z_i \mid i \in \sI\}$ determines $f$; thus $\Ddt(f) \leq |\sI| < 10c$. 
Towards a contradiction, assume that there exist $z^{(1)}, z^{(2)} \in \sS \cap \{z: \forall i \in \sI, z_i=a_i\}$ such that $f(z^{(1)}) \neq f(z^{(2)})$. We will construct a string $y=(y_1, \ldots, y_n) \in \{0,1\}^{nb}$ in the following way:
\begin{description}
\item[$i \in \sI$]: Choose $y_i$ such that $\mathsf{IP}(y_i, x^{(1)}_i)=\mathsf{IP}(y_i, x^{(2)}_i)=a_i$.
\item[$i \notin \sI$]: Choose $y_i$ such that $\mathsf{IP}(y_i, x^{(1)}_i)=z^{(1)}_i$ and $\mathsf{IP}(y_i, x^{(2)}_i)=z^{(2)}_i$.
\end{description}
Note that we can always choose a $y$ as above since for each $i \in [n]$, $x^{(1)}_i, x^{(2)}_i \neq 0^b$, and for each $i \notin \sI$, $x^{(1)}_i \neq x^{(2)}_i$. By the above construction, $f \circ \IP(x^{(1)}, y)=f(z^{(1)})$ and $f \circ \IP(x^{(2)}, y)=f(z^{(2)})$. Since by assumption $f(z^{(1)}) \neq f(z^{(2)})$, we have that $f \circ \IP(x^{(1)}, y) \neq f \circ \IP(x^{(2)}, y)$. But since Alice sends the same message on inputs $x^{(1)}$ and $x^{(2)}$, $\Pi$ produces the same output on $(x^{(1)}, y)$ and $(x^{(2)}, y)$. This contradicts the correctness of $\Pi$.
\end{proof}
\begin{remark}
It can be seen that the proof of Theorem~\ref{thm:detlift} also works when the inner gadget $g : \zone^{b_1} \times \zone^{b_2} \to \zone$ satisfies the following general property: There exists a subset $X$ of $\zone^{b_1}$ (Alice's input in the gadget) such that:
\begin{itemize}
    \item $|X| \geq 3$,
    \item for all $x_1 \neq x_2 \in X$ and all $b_1, b_2 \in \zone$, there exists $y \in \zone^{b_2}$ such that $g(x_1, y) = b_1$ and $g(x_2, y) = b_2$.
\end{itemize}
This is satisfied, for example, for the Addressing function on $\zone^{\log b + b}$ when $b \geq 4$ (see Definition~\ref{defn: addressing}). For $g = \IP_b$, the set $X$ equals $\zone^b \setminus \cbra{0^b}$.
\end{remark}

\section{Composition with AND}\label{sec: naandt is not communication}

We first investigate the relationship between \textnaandt~complexity and M\"obius sparsity of Boolean functions.
Recall that Claim~\ref{claim: trivial bounds} shows that for all Boolean functions $f : \zone^n \to \zone$, $\log \spar(f) \leq \NAANDT(f) \leq \spar(f)$.
A natural question to ask is whether both of the bounds are tight, i.e.~are there Boolean functions witnessing tightness of each bound? The first bound is trivially tight for any Boolean function with full M\"obius sparsity, for example, the NOR function: querying all the input bits (which is querying $n$ many ANDs) immediately yields the value of the function, and its M\"obius sparsity can be shown to be $2^n$.
One might expect that the upper bound is not tight in view of Theorem~\ref{thm: san19}. The Addressing function witnesses tightness of the quadratic gap in Theorem~\ref{thm: san19}. This gives rise to the natural question of whether an analogous bound holds true in the M\"obius-world: is it true for all Boolean functions $f$ that $\NAANDT(f) = \widetilde{O}(\sqrt{\spar(f)})$?
Interestingly we show in Appendix~\ref{app: addressing} that the Addressing function already gives a negative answer to this question. 
In this section we first observe that a stronger separation holds, and there exists a function ($\OMB_n$) for which the second inequality in Claim~\ref{claim: trivial bounds} is in fact an equality. We then use this same function to prove Theorem~\ref{thm: intro omb separation}, which gives a maximal separation between $\QNAANDT(f)$ and $\Doneway(f \circ \AND_2)$. Finally, we prove Theorem~\ref{thm: intro symmetric naandt communication}, which says that $\NAANDT(f)$ is at most quadratically large in $\Dcc(f \circ \AND)$ for symmetric $f$.

\subsection{Pattern Complexity and One-Way Communication Complexity}\label{sec: pattern complexity}

In this section we observe that the logarithm of the pattern complexity, $\mpat(f)$, of a Boolean function $f$ equals the deterministic one-way communication complexity of $f \circ \AND$. We also give bounds on $\NAANDT(f)$ in terms of $\mpat(f)$.
As a consequence we also show that $\Doneway(f \circ \AND) \geq \log(\NAANDT(f))$.
\begin{claim}\label{claim: doneway equals log pattern}
Let $f : \zone^n \to \zone$ be a Boolean function. Then $
\Doneway(f \circ \AND) = \lc\log(\mpat(f))\rc$.
\end{claim}

\begin{proof}
Write the M\"obius expansion of $f$ as
\begin{align}\label{eqn: mobius expansion of f}
    f = \sum_{S \in \S_f}\wh{f}(S)\AND_S.
\end{align}
Say $\mpat(f) = k$. We first show that $\Doneway(f \circ \AND) \leq \lc\log k\rc$ by exhibiting a one-way protocol of cost $\lc\log k\rc$. Alice computes the pattern of $x$ and sends Bob the pattern using $\lc\log k\rc$ bits of communication. Bob now knows the values of $\cbra{\AND_S(x) : S \in \S_f}$. Since Bob can compute $\cbra{\AND_S(y) : S \in \S_f}$ without any communication, he can now compute the value of $f \circ \AND(x,y)$ using the formula
\[
(f \circ \AND)(x,y) = \sum_{S \in \S_f}\wh{f}(S)\AND_S(x)\AND_S(y).
\]

It remains to show that $\Doneway(f \circ \AND) \geq \lc\log k\rc$. Let $\Doneway(f \circ \AND) = d$. Thus there are at most $2^d$ messages that Alice can send Bob. We show that any two inputs $x, x' \in \zone^n$ for which Alice sends the same message have the same pattern, which would prove $2^d \geq k$, and prove the claim since $d$ must be an integer.

Let $x, x'$ be 2 inputs to Alice for which her message to Bob is $m$. We have
\begin{align*}
    (f \circ \AND)(x,y) & = \sum_{S \in \S_f}\wh{f}(S)\AND_S(x)\AND_S(y)\\
    (f \circ \AND)(x',y) & = \sum_{S \in \S_f}\wh{f}(S)\AND_S(x')\AND_S(y)
\end{align*}
Since $m$  and $y$ completely determine the value of the function, we must have \begin{align*}
\sum_{S \in \S_f}\wh{f}(S)\AND_S(x)\AND_S(y) = \sum_{S \in \S_f}\wh{f}(S)\AND_S(x')\AND_S(y) \qquad\text{for all $y \in \zone^n$.}
\end{align*}
Define the functions $g_x, g_{x'} : \zone^n \to \zone$ by
\begin{align*}
    g_x(y) & = \sum_{S \in \S_f}\wh{f}(S)\AND_S(x)\AND_S(y)\\
    g_{x'}(y) & = \sum_{S \in \S_f}\wh{f}(S)\AND_S(x')\AND_S(y).
\end{align*}
Thus by uniqueness of the M\"obius expansion of Boolean functions, $g_x = g_{x'}$ as functions of $y$. This implies $\wh{g_x}(S) = \wh{g_{x'}}(S)$ for all $S \in \S_f$. Since $\wh{g_x}(S) = \wh{f}(S)\AND_S(x)$ and $\wh{g_{x'}}(S) = \wh{f}(S)\AND_S(x')$ for all $S \in \S_f$,
\begin{align*}
\AND_S(x) = \AND_S(x') \qquad \textnormal{for all}~S \in \S_f,
\end{align*}
This shows that the pattern induced by $x$ and the pattern induced by $x'$ are the same, concluding the proof.
\end{proof}

Next we show that the pattern complexity of $f$ is bounded below by the M\"obius sparsity of $f$. 


\begin{claim}\label{claim: patcomp at least spar}
Let $f : \zone^n \to \zone$ be a Boolean function. Then $\mpat(f) \geq \spar(f)$.
\end{claim}

\begin{proof}
Recall that $\S_f$ denotes the M\"obius support of $f$. For each $S \in \S_f$, define the input $x^S$ to be the $n$-bit characteristic vector of the set $S$. We now show that each of these inputs induces a different pattern for $f$. Let $S_1 \neq S_2 \in \S_f$, with $|S_1| \geq |S_2|$. Since they are different sets, there must be an index $j \in S_1$ such that $j \notin S_2$. Note that $\AND_{S_1}(x^{S_1}) = 1$. On the other hand $x^{S_2}_j = 0$ implies $\AND_{S_1}(x^{S_2}) = 0$. Hence $x^{S_1}$ and $x^{S_2}$ induce different patterns. Since $\spar(f) = \abs{\S_f}$, this completes the proof.
\end{proof}

From Claim~\ref{claim: trivial bounds} we know that $\spar(f) \geq \NAANDT(f)$ and from Claim~\ref{claim: doneway equals log pattern} we know that $\Doneway(f \circ \AND) = \lc\log(\mpat(f))\rc$. Along with Claim~\ref{claim: patcomp at least spar}, these imply the following claim.

\begin{claim}\label{claim: doneway vs naandt}
Let $f : \zone^n \to \zone$ be a Boolean function. Then $\lc\log(\NAANDT(f))\rc \leq \Doneway(f \circ \AND) \leq \NAANDT(f)$.
\end{claim}

\begin{proof}
For the upper bound on $\Doneway(f \circ \AND)$, let $\S = \bra{S_1, \dots, S_k}$ be an NAADT basis for $f$. By Claim~\ref{claim: naandt basis monomial product}, every monomial in the M\"obius support of $f$ is a product of some of these $\AND_{S_i}$'s. Since there are at most $2^k$ possible values for $\bra{\AND_{S_i}(x) : i \in [k]}$ and since these completely determine the pattern of $x$ for any given $x \in \zone^n$, we have
\[
\mpat(f) \leq 2^{\NAANDT(f)},
\]
which proves the required upper bound in view of Claim~\ref{claim: doneway equals log pattern}.

For the lower bound, we have
\[
\Doneway(f \circ \AND) = \lc\log(\mpat(f))\rc \geq \lc\log(\spar(f))\rc \geq \lc\log(\NAANDT(f))\rc,
\]
where the equality follows from Claim~\ref{claim: doneway equals log pattern}, the first inequality follows from Claim~\ref{claim: patcomp at least spar} and the last inequality follows from Claim~\ref{claim: trivial bounds}.
\end{proof}

The pattern complexity of $f$ is trivially at most $2^{\spar(f)}$ since each pattern is a $\spar(f)$-bit string. Interestingly we show that there is no function for which this bound is tight. 
\begin{claim}\label{claim: patcomp is never full}
    Let $f: \zone^n \rightarrow \zone$ be a Boolean function. Then $\mpat(f) \leq 2^{(1-\Omega(1))\spar(f)}$.
\end{claim}

In fact, we obtain an explicit constant that yields an upper bound of $2^{((\log 6)/3)\spar(f) + 1} \approx 2^{0.86\spar(f) + 1}$.
Our proof of Claim~\ref{claim: patcomp is never full} relies on the following observation about the structure of the M\"obius support of any Boolean function.

\begin{claim}\label{clm: mobius support structure}
    Let $f: \zone^n \to \zone$ be a Boolean function with M\"obius support $\S_f$. For any two distinct sets $S,T \in \S_f$ there exists a set of `partners' $p(\cbra{S, T}) \subseteq \S_f$ such that
    \begin{itemize}
        \item $p(\cbra{S, T}) \neq \cbra{S,T}$,
        \item $\abs{p(\cbra{S, T})} = 2$ if $S \cup T \notin \S_f$ and $\abs{p(\cbra{S, T})} = 1$ if $S \cup T \in \S_f$, and
        \item $\bigcup_{U \in p(\cbra{S, T})} U= S \cup T$.
    \end{itemize}
\end{claim}

\begin{proof}
Let $\sum_{S \in \S_f} \wh{f}(S) \AND_{S}$ be the M\"obius expansion of $f$. Since $f$ has range $\zone$, we know that $f = f^2$. However, $$f^2 = \left( \sum_{S \in \S_f} \wh{f}(S) \AND_{S} \right) \left( \sum_{T \in \S_f} \wh{f}(T) \AND_{T} \right) =  \sum_{W \subseteq [n]} \left( \sum_{S,T \subseteq [n] : S \cup T = W} \wh{f}(S) \wh{f}(T) \right) \AND_W.$$
Since the M\"obius expansion of $f$ is unique, we can compare the two expansions to see that for all sets $W \subseteq [n]$, 
\begin{equation}\label{eqn: mobius titsworth}
    \wh{f}(W) = \sum_{S,T \subseteq [n] : S \cup T = W} \wh{f}(S) \wh{f}(T).    
\end{equation}
As a consequence we have the following structure. Let $S \neq T \in \S_f$ such that $S \cup T \notin \S_f$. Since $\wh{f}(S \cup T) = 0$, the summation corresponding to $W = S \cup T$ in Equation~\eqref{eqn: mobius titsworth} must have at least one non-zero summand apart from $\wh{f}(S) \wh{f}(T)$. Hence there must exist $U \neq V \in \S_f$ such that $\cbra{S,T} \neq \cbra{U,V}$ and $U \cup V = S \cup T$. We choose an arbitrary such pair $\cbra{U,V}$ and define $p(\cbra{S,T}) = \cbra{U,V}$. For $S,T \in \S_f$ such that $S \cup T \in \S_f$, let $p(\cbra{S,T})$ be defined as $\cbra{S \cup T}$. It clearly satisfies the necessary conditions.
\end{proof}

\begin{observation}\label{obs: mobius partners patterns}
Let $f: \zone^n \to \zone$ be a Boolean function with M\"obius support $\S_f$. For any two distinct sets $S,T \in \S_f$, let $p(\cbra{S, T}) \subseteq \S_f$ be as in Claim~\ref{clm: mobius support structure}. Then for any pattern $P \in \zone^{\mathcal{S}_f}$,
\[
P_S \cdot P_T = \prod_{W \in p(\cbra{S, T})} P_W.
\]
\end{observation}
\begin{proof}
Let $P$ be a pattern in $\zone^{\S_f}$. There must exist an $x \in \zone^n$ such that for all sets $W \in \S_f$, $P_W = \AND_W(x)$. Since $S \cup T = \bigcup_{W \in p(\cbra{S, T})} W$, we have $P_S \cdot P_T = \AND_{S \cup T}(x) = \prod_{W \in p(\cbra{S, T})} P_W$.
\end{proof}

\begin{proof}[Proof of Claim~\ref{claim: patcomp is never full}]
We analyze the pattern complexity of $f$ in iterations. To define these iterations, we define a sequence of subsets of $\S_f$, described in Algorithm~\ref{algo: iteration generation}.

\begin{algorithm}[H]\label{algo: iteration generation}
\SetAlgoLined
\textbf{Initialize } $\mathcal{T}_0 \gets \emptyset, i \gets 0$. \\
\While{$\abs{\mathcal{T}_i} \leq \spar(f)-2$}{
    Choose $S, T$ with $S \neq T$ from $\S_f \setminus \mathcal{T}_i$. \\
    Set $\mathcal{T}_{i+1} \gets \mathcal{T}_i \cup \cbra{S,T} \cup p(\cbra{S,T})$. \\
    Set $i \gets i+1$.
}
Set $\mathsf{num\_iterations} \gets i$. \\
Set $\mathcal{T}_{\mathsf{num\_iterations}+1} \gets \S_f$.
\caption{Defining the Iterations}
\end{algorithm}

For $i \in \cbra{0,\dots,\mathsf{num\_iterations}+1}$, define the partial patterns $$\mathcal{P}_i := \cbra{P \in \zone^{\mathcal{T}_i} : P = \rbra{\AND_S(x)}_{S \in \mathcal{T}_i}~\textnormal{for some}~x \in \zone^n}.$$ We will now prove that
\begin{equation}\label{eqn: partial patterns bound}
    \forall j \in \cbra{0,\dots,\mathsf{num\_iterations}},~ \abs{\mathcal{P}_j} \leq 2^{\frac{\log 6}{3}\abs{\mathcal{T}_j}}.
\end{equation}
We prove this by induction. In the remainder of this proof, let $\alpha := \frac{\log 6}{3} \approx 0.86$.
Equation~\eqref{eqn: partial patterns bound} is true when $j=0$ since both sides are $1$. Now let $i>0$ and assume as our induction hypothesis that Equation~\eqref{eqn: partial patterns bound} is true when $j=i-1$. As our inductive step, we will prove that for every partial pattern $P \in \mathcal{P}_{i-1}$, the number of partial patterns $Q \in \mathcal{P}_i$ that extend $P$ (in the sense that $Q$ restricted to indices in $\mathcal{T}_{i-1}$ is equal to $P$) is at most $2^{\alpha(\abs{\mathcal{T}_i} - \abs{\mathcal{T}_{i-1}})}$. Since every partial pattern in $\mathcal{P}_i$ is an extension of a partial pattern in $\mathcal{P}_{i-1}$, this would imply that $\abs{\mathcal{P}_i} \leq 2^{\alpha(\abs{\mathcal{T}_i} - \abs{\mathcal{T}_{i-1}})}\abs{\mathcal{P}_{i-1}}$. Along with our induction hypothesis, this will prove Equation~\eqref{eqn: partial patterns bound} for $j=i$, and hence for all $j$.

Our proof of the inductive step will involve some case analysis. Fix a partial pattern $P \in \mathcal{T}_{i-1}$. Let $S,T$ be the sets chosen when constructing $\mathcal{T}_i$ from $\mathcal{T}_{i-1}$. We know from Observation~\ref{obs: mobius partners patterns} that any partial pattern $Q \in \mathcal{P}_i$ must satisfy $Q_{S} \cdot Q_{T} = \prod_{W \in p(\cbra{S,T})} Q_{W}$. We split the analysis into cases based on whether or not $S \cup T \in \S_f$ and also based on the number of sets in $\mathcal{T}_i \setminus \mathcal{T}_{i-1}$.

We now show the inductive step, that is, for all $i \in \cbra{1,\dots,\mathsf{num\_iterations}}$, the number of extensions of any partial pattern $P \in \mathcal{P}_{i - 1}$ to a partial pattern $Q$ in $\mathcal{P}_{i}$ is at most $2^{\alpha(\abs{\mathcal{T}_i} - \abs{\mathcal{T}_{i-1}})}$.

\begin{itemize}
    \item Suppose $S \cup T \notin \S_f$. In this case we know $\abs{p(\cbra{S,T})} = 2$. Let $p(\cbra{S,T}) = \cbra{U,V}$. Consider the set $\mathcal{W} := \mathcal{T}_{i} \setminus \mathcal{T}_{i-1} = (\cbra{U, V} \cup \cbra{S,T}) \setminus \mathcal{T}_{i-1}$.
    \begin{itemize}
        \item If $\abs{\mathcal{W}} = 2$: This can happen either because $\cbra{U, V} \subseteq \mathcal{T}_{i-1}$ or because one of $U$ and $V$ is in $\mathcal{T}_{i-1}$ and the other is in $\cbra{S,T}$. We analyze these subcases below.
        \begin{itemize}
            \item If $\cbra{U,V} \subseteq \mathcal{T}_{i-1}$, then $P$ already determines $Q_{U}$ and $Q_{V}$, and hence $Q_{U} \cdot Q_{V}$. One of the following two cases hold. The third implication in both of the cases below use Observation~\ref{obs: mobius partners patterns}.
            \begin{align*}
                (P_{U}, P_{V}) = (1, 1) & \implies (Q_{U}, Q_{V}) = (1, 1) \implies Q_{U} \cdot Q_{V} = 1 \\
                & \implies Q_{S} \cdot Q_{T} = 1 \implies (Q_{S}, Q_{T}) = (1, 1)\\
                & \textnormal{or}\\
                (P_{U}, P_{V}) \in \cbra{(0,0), (0,1), (1,0)} & \implies (Q_{U}, Q_{V}) \in \cbra{(0,0), (0,1), (1,0)} \\ & \implies Q_{U} \cdot Q_{V} = 0 \implies Q_{S} \cdot Q_{T} = 0 \\ 
                & \implies (Q_{S}, Q_{T}) \in \cbra{(0,0), (0,1), (1,0)}.
            \end{align*}
            \item One of $U$ or $V$ is in $\mathcal{T}_{i-1}$ and the other is in $\cbra{S,T}$. Without loss of generality assume $U \in \mathcal{T}_{i-1}$ (which implies $Q(U) = P(U)$) and $V=T$. Thus $S \cup T = U \cup T$, and Observation~\ref{obs: mobius partners patterns} implies $P_{U} \cdot Q_T = Q_S \cdot Q_T$. If $P_{U} = 0$ then
            \begin{align*}
                P_{U} \cdot Q_{T} = 0 \implies Q_{S} \cdot Q_{T} = 0 \implies (Q_{S}, Q_{T}) \in \cbra{(0,0), (0,1), (1,0)}.
            \end{align*}
            If $P_{U} = 1$ then
            \begin{align*}
                Q_T = Q_S \cdot Q_T \implies (Q_{S}, Q_{T}) \in \cbra{(0, 0), (1,0), (1,1)}.
            \end{align*}
        \end{itemize}
        
        In both of the above subcases, $P$ has at most $3$ extensions in $\mathcal{P}_{i}$. Note that $3 < 2^{2\alpha} \approx 3.3$, so our inductive step holds in this case.
        
        \item If $\abs{\mathcal{W}} = 3$: Without loss of generality, let $\mathcal{W} = \cbra{S,T,U}$. This case can arise either because $V \in \mathcal{T}_{i-1}$ or because $V \in \cbra{S,T}$. Again, we analyze both subcases.
        \begin{itemize}
            \item If $V \in \mathcal{T}_{i-1}$, then $Q_{V} = P_{V}$ and thus Observation~\ref{obs: mobius partners patterns} implies $Q_{U} \cdot P_{V} = Q_{S} \cdot Q_{T}$. If $P_{V}=0$ then
            \begin{align*}
                &Q_{U} \cdot P_{V} = 0 \implies Q_S \cdot Q_T = 0\\
                \implies &(Q_{S}, Q_{T}, Q_{U}) \in \cbra{(0,0,0), (0,0,1), (0,1,0), (0,1,1), (1,0,0), (1,0,1)}.
            \end{align*}
            If $P_{V}=1$ then
            \begin{align*}
                Q_{U} = Q_S \cdot Q_T \implies (Q_{S}, Q_{T}, Q_{U}) \in \cbra{(0,0,0), (0,1,0), (1,0,0), (1,1,1)}.
            \end{align*}
            \item If $V \in \cbra{S,T}$ (without loss of generality assume $V=T$), then Observation~\ref{obs: mobius partners patterns} implies that
            \begin{align*}
                &Q_{U} \cdot Q_T = Q_S \cdot Q_T \\ \implies &(Q_{S}, Q_{T}, Q_{U}) \in \cbra{(0,0,0), (0,0,1), (0,1,0), (1,0,0), (1,0,1), (1,1,1)}.
            \end{align*}
        \end{itemize}
        
        In both of these subcases, $P$ had at most $6$ extensions in $\mathcal{P}_i$. Since $2^{3\alpha} = 6$, the inductive step holds in this case.
        
        \item If $\abs{\mathcal{W}} = 4$, Observation~\ref{obs: mobius partners patterns} implies that
        \begin{align*}
            &Q_{U} \cdot Q_{V} = Q_{S} \cdot Q_{T} \\ \implies & (Q_{S},Q_{T},Q_{U},Q_{V}) \in \rbra{\cbra{(0,0),(0,1),(1,0)} \times \cbra{(0,0),(0,1),(1,0)}} \cup \cbra{(1,1,1,1)}.
        \end{align*}
        Thus there are 10 possibilities for $(Q_S,Q_T,Q_{U},Q_{V})$, so $P$ has at most $10$ extensions in $\mathcal{P}_i$. Since $10 < 2^{4\alpha} \approx 10.9$, the induction step holds in this case as well.
    \end{itemize}
    \item Suppose $S \cup T \in \S_f$. In this case $p(\cbra{S,T}) = \cbra{S \cup T}$. Again, consider the set $\mathcal{W} := \mathcal{T}_{i} \setminus \mathcal{T}_{i-1} = (\cbra{S \cup T} \cup \cbra{S,T}) \setminus \mathcal{T}_{i-1}$.
    \begin{itemize}
        \item If $\abs{\mathcal{W}} = 2$: This can arise either because $S \cup T \in \mathcal{T}_{i-1}$ or because $S \cup T \in \cbra{S,T}$. We analyze both of these subcases below.
        \begin{itemize}
            \item If $S \cup T \in \mathcal{T}_{i-1}$, Observation~\ref{obs: mobius partners patterns} implies that $P_{S \cup T} = Q_S \cdot Q_T$. If $P_{S \cup T}=1$, then
            \begin{align*}
                Q_S \cdot Q_T = 1 \implies (Q_S,Q_T) = (1,1).
            \end{align*}
            If $P_{S \cup T}=0$, then
            \begin{align*}
                Q_S \cdot Q_T = 0 \implies (Q_S,Q_T) \in \cbra{(0,0),(0,1),(1,0)}.
            \end{align*}
            \item If $S \cup T \in \cbra{S,T}$ (without loss of generality assume $S \cup T = S$), then Observation~\ref{obs: mobius partners patterns} implies that $Q_S = Q_S \cdot Q_T$. Therefore
            \begin{align*}
                (Q_S,Q_T) \in \cbra{(0,0),(0,1),(1,1)}.
            \end{align*}
        \end{itemize}
        In both of the above subcases, $P$ has at most $3$ extensions in $\mathcal{P}_i$. Since $3 < 2^{2\alpha} \approx 3.3$, the inductive step holds in this case.
        \item If $\abs{\mathcal{W}} = 3$, Observation~\ref{obs: mobius partners patterns} implies
        \begin{align*}
            Q_{S \cup T} = Q_S \cdot Q_T \implies (Q_S,Q_T,Q_{S \cup T}) \in \cbra{(0,0,0),(0,1,0),(1,0,0),(1,1,1)}.
        \end{align*}
        In this case $P$ has at most $4$ extensions in $\mathcal{P}_i$. Since $4 < 2^{3\alpha} = 6$, the inductive step holds here as well.
    \end{itemize}
\end{itemize}

Since the inductive step holds in all the cases, we conclude that Equation~\eqref{eqn: partial patterns bound} is true.
Finally, note that the while loop in Algorithm~\ref{algo: iteration generation} quits when $\abs{\mathcal{T}_i} \geq \spar(f)-1$. If it quits with $\abs{\mathcal{T}_i} = \spar(f)$, then Equation~\eqref{eqn: partial patterns bound} implies that $\mpat(f) \leq 2^{\alpha\cdot\spar(f)}$. If it quits with $\abs{\mathcal{T}_i} = \spar(f)-1$, then each of the partial patterns in $\mathcal{P}_i$ can have at most two extensions to actual patterns of $f$. Hence even in this case $\mpat(f) \leq 2^{\alpha\cdot\spar(f)+1}$.
\end{proof}

The analogous relation to Claim~\ref{claim: patcomp is never full} in the Fourier-world has been nearly determined by Sanyal~\cite{San19}; their main result (Theorem~\ref{thm: san19}) essentially shows that the Fourier analog of pattern complexity of a Boolean function is at most exponential in the square root of its Fourier sparsity. This is a stronger bound than that in Claim~\ref{claim: patcomp is never full}, but the same bound cannot hold in the M\"obius-world  since the Addressing function witnesses $\mpat(\ADDR_n) \geq 2^{\spar(\ADDR_n)^{\log_32}}$ (see Claims~\ref{claim: addr sparsity} and~\ref{claim: addr patcomp lower bound}). Nevertheless we conjecture that a stronger bound is possible.
\begin{conjecture}\label{conj: patcomp ub}
Let $f: \zone^n \rightarrow \zone$ be a Boolean function. Then $\mpat(f) \leq 2^{\rbra{\spar(f)^{1 - \Omega(1)}}}$.
\end{conjecture}
Conjecture~\ref{conj: patcomp ub} would strengthen Theorem~\ref{thm: 1way sublinear rank}, showing that $\Doneway(f \circ \AND) = \rank(M_{f \circ \AND})^{1 - \Omega(1)}$.
\begin{proof}[Proof of Theorem~\ref{thm: 1way sublinear rank}]
We have
\[
    \Doneway(f \circ \AND) = \lc\log(\mpat(f))\rc \leq (1-\Omega(1))\spar(f) \leq (1-\Omega(1))\rank(M_{f \circ \AND}),
\]
where the equality holds by Claim~\ref{claim: doneway equals log pattern}, the first inequality follows from Claim~\ref{claim: patcomp is never full} and the last inequality holds by Equation~\eqref{eqn: mobius sparsity rank}.
\end{proof}
Our results regarding the one-way communication complexity of $f \circ \AND$ use the Booleanness of $f$ to bring out mathematical insights about the dependencies of monomials in the M\"obius support of $f$. These dependencies enable us to establish interesting bounds on the pattern complexity of $f$. We hope that pattern complexity will find more use in future research.



\subsection{Deterministic Complexity}
We prove in this section that the \textnaandt~complexity of $\OMB_n$ is maximal whereas the one-way communication complexity of $\OMB_n \circ \AND$ is small.

\begin{claim}\label{claim: max naandt OMB}
Let $n$ be a positive integer. Then $\NAANDT(\OMB_n) = n$. Moreover, $\spar(\OMB_n) = n$ if $n$ is even, and $\spar(\OMB_n) = n+1$ if $n$ is odd.
\end{claim}

\begin{proof}[Proof of Claim~\ref{claim: max naandt OMB}]
Write the polynomial representation of $\OMB_n$ as $\OMB_n(x) = $
\begin{align}\label{eqn: OMB and expansion}
(1-x_n)\cdot 0 + x_n(1-x_{n-1}) \cdot 1 + x_nx_{n-1}\OMB_{n-2}(x_1, \dots, x_{n-2}) & \quad \text{if $n$ is even, or}\\
(1-x_n)\cdot 1 + x_n(1-x_{n-1}) \cdot 0 + x_nx_{n-1}\OMB_{n-2}(x_1, \dots, x_{n-2}) & \quad \text{if $n$ is odd}.
\end{align}
The M\"obius support of $\OMB_n$ equals $\bra{\bra{j, \dots, n} : j \leq n} \cup \bra{\emptyset}$ if $n$ is odd, and $\bra{\bra{j, \dots, n} : j \leq n}$ if $n$ is even. Thus $\spar(\OMB_n) = n+1$ if $n$ is odd, and equals $n$ if $n$ is even.

We now show that the $\NAANDT(\OMB_n) = n$. Let $\S$ denote a NAADT basis for $\OMB_n$.
By Claim~\ref{claim: naandt basis monomial product}, any monomial in the M\"obius expansion of $\OMB_n$ can be expressed as a product of some ANDs from $\S$. 
Thus, $\bra{n}$ must participate in $\S$ since it appears in its M\"obius support. Next, since $\bra{n-1, n}$ appears in the support as well, either $\bra{n-1, n}$ or $\bra{n-1}$ must appear in $\S$. Continuing iteratively, we conclude that for all $i \in [n]$, there must exist a set in $\S$ that contains $i$, but does not contain any $j$ for $j < i$. This implies that $|\S| \geq n$. Equality holds since $\NAANDT(f) \leq n$ for any Boolean function $f : \zone^n \to \zone$.
\end{proof}
Thus $\OMB_n$ witnesses that \textnaandt~complexity can be as large as sparsity. 
We remark here that $\OMB_n$ admits a simple (adaptive) AND-decision tree that makes $O(\log n)$ AND-queries in the worst case. This uses a binary search using AND-queries to determine the right-most index where a 0 is present.
One might expect that a result similar to Claim~\ref{claim: foxor oneway equals napdt} holds when the inner function is $\AND$ instead of $\XOR$. That is, it is plausible that the deterministic one-way communication complexity of $f \circ \AND$ equals the \textnaandt~complexity of $f$. We show that this is not true, and exhibit an exponential separation between $\Doneway(\OMB_n \circ \AND)$ and $\NAANDT(\OMB_n)$.

\begin{claim}\label{claim: doneway omb}
Let $n$ be a positive integer. Then $\Doneway(\OMB_n \circ \AND) = \lc\log(n + 1)\rc$.
\end{claim}

\begin{proof}
By expanding Equation~\eqref{eqn: OMB and expansion} we have that the M\"obius support of $\OMB_n$ equals the set $\S = \cbra{\bra{n}, \bra{n-1, n}, \dots, \bra{n, n-1, \dots, 1}}$ if $n$ is an even integer, and equals
the set $\S = \cbra{\emptyset, \bra{n}, \bra{n-1, n}, \dots, \bra{n, n-1, \dots, 1}}$ if $n$ is an odd integer.
It is easy to verify that the only possible M\"obius patterns attainable (ignoring the empty set since it always evaluates to 1) are $1^i0^{n-i}$, for $i \in \bra{0, 1, \dots, n}$. Moreover, all of these patterns are attainable: the pattern $1^i0^{n-i}$ is attained by the input string $0^{n-i}1^i$. Thus $\mpat(\OMB_n) = n+1$.
Claim~\ref{claim: doneway equals log pattern} implies $\Doneway(\OMB_n \circ \AND) = \lc\log(n+1)\rc$.
\end{proof}
We obtain our main result of this section, which follows from Claim~\ref{claim: max naandt OMB} and Claim~\ref{claim: doneway omb}.
\begin{theorem}\label{thm: omb and exponential separation}
Let $n$ be a positive integer. Then $    \NAANDT(\OMB_n) = n$ and $\Doneway(\OMB_n \circ \AND) = \lc\log(n + 1)\rc$.
\end{theorem}

\subsection{Quantum Complexity}

We prove that even the \textqnaandt~complexity of $\OMB_n$ is $\Omega(n)$. In view of the small one-way communication complexity of $\OMB_n \circ \AND$ from Claim~\ref{claim: doneway omb}, Theorem~\ref{thm: intro omb separation} then follows.
\begin{theorem}\label{thm: quantum NAANDT OMB}
Let $n$ be a positive integer. Then $\QNAANDT(\OMB_n) = \Omega(n)$.
\end{theorem}

\begin{remark}The fact that the \textnaandt~complexity (in fact even the non-adaptive \emph{monotone} decision tree complexity, where the tree is allowed to make arbitrary monotone queries rather than just ANDs) of $\OMB_n$ equals $\Omega(n)$ already follows by a recent result of Amireddy, Jayasurya and Sarma~\cite{AJS20}. They show that any function with alternating number $k$ (which is the largest number of switches of the function's value along a monotone path from $0^n$ to $1^n$) must have non-adaptive monotone decision tree complexity at least $k$. In Theorem~\ref{thm: quantum NAANDT OMB}, we show that even the \emph{quantum} \textnaandt~complexity of $\OMB_n$ is $\Omega(n)$. In fact our proof can easily be adapted to show that any function with alternating number $k$ must have quantum non-adaptive monotone decision tree complexity $\Omega(k)$, thus generalizing the above-mentioned result~\cite{AJS20}.
\end{remark}

Before we prove this theorem, we introduce an auxiliary function and state some properties of it that are of use to us.
\begin{definition}\label{defn: omb'}
Let $n$ be a positive integer. Let us define the set $\sS \subset \zone^n$ to be $\sS = \cbra{x \in \zone^n : x = 0^i1^{n - i}~\text{for some}~i \in [n]}$. Define the partial function $\OMB'_n : \sS \to \zone$ by $\OMB'_n(x) = \OMB_n(x)$.
\end{definition}

\begin{claim}\label{claim: omb' reduction}
Let $n$ be a positive integer. Then $\RNAANDT(\OMB'_n) = \Rdt(\OMB'_n)$ and $\QNAANDT(\OMB'_n) = \Qdt(\OMB'_n)$.
\end{claim}
\begin{proof}
The inequalities $\RNAANDT(\OMB'_n) \leq \Rdt(\OMB'_n)$ and $\QNAANDT(\OMB'_n) \leq \Qdt(\OMB'_n)$ are clear by definitions. In the other direction, consider a (randomized or quantum) \textnaandt~computing $\OMB'_n$. Let $S \subseteq [n]$ be a subset of variables whose AND is queried (possibly as part of a superposition) in this tree. Let $S = \cbra{i_1, \dots, i_k}$, where $i_1 < i_2 < \cdots < i_k$. Observe that $\AND_{S}(x) = x_{i_1}$ for all $x$ in the domain of $\OMB'_n$. Thus the query to $\AND_S$ can be replaced by a query to the $i_1$'th variable. Hence any (randomized or quantum) \textnaandt~can be converted to a (randomized or quantum) non-adaptive decision tree without changing its correctness or complexity. Thus $\RNAANDT(\OMB'_n) \geq \Rdt(\OMB'_n)$ and $\QNAANDT(\OMB'_n) \geq \Qdt(\OMB'_n)$, proving the claim.
\end{proof}

We require the following result, which follows implicitly from a result of Montanaro~\cite{Mon10}. 

\begin{theorem}\label{thm: quantum dt lower bound}
Let $\sS \subseteq \zone^n$, $I \subseteq [n]$ and $f : \sS \to \zone$ be such that for all $i \in I$ there exists $x \in \sS$ such that $f(x \oplus e_i) = 1 - f(x)$. Then $\Qdt(f) = \Omega(\abs{I})$.
\end{theorem}
We provide a proof below for completeness.
\begin{proof}
Consider a $T$-query non-adaptive quantum query algorithm that computes $f$ to error at most $\eps = 1/3$. The algorithm works with a state space $\ket{a_1, \dots, a_T}\ket{b}\ket{w}$, where each $a_j \in [n]$, $b \in \zone^T$ and the last register captures the workspace. The algorithm does the following on input $x \in S$:
\begin{itemize}
    \item It starts in an input-independent state, say $\sum_{a_1, \dots, a_T, b, w}\alpha_{a_1, \dots, a_T, b, w}\ket{a_1, \dots, a_T}\ket{b}\ket{w}$,
    \item applies $T$ non-adaptive queries, that act on the basis states as follows:
    \[
    O_x : \ket{a_1, \dots, a_T}\ket{b_1, \dots, b_T}\ket{w} \mapsto \ket{a_1, \dots, a_T}\ket{b_1 \oplus x_{a_1}, \dots, b_T \oplus x_{a_T}}\ket{w}.
    \]
    \item It then applies a two-outcome projective measurement $\cbra{\Pi, I - \Pi}$ on the resulting state $O_x \ket{\psi}$ and outputs a value depending on the measurement outcome.
\end{itemize}

Let $\ket{\psi_x}$ denote the state of the algorithm on input $x \in \zone^n$ just before measurement.

Fix $i \in I$. By assumption, there exists $x \in \zone^n$ such that $f(x) = 1 - f(x \oplus e_i)$. We define 
\[
J_i := \cbra{(a_1,\dots,a_T,b,w) \mid i \in \cbra{a_1,\dots,a_T}} 
\]
to be (the indices of) those basis states that query $x_i$. Write the initial state of our algorithm as $\ket{\psi} = \ket{\phi_{1i}} + \ket{\phi_{2i}}$, where $\ket{\phi_{1i}} = \sum_{s \in J_i} \alpha_s \ket{s}$. Thus,
\begin{align*}
    0.114 \approx 2 - 4\sqrt{2}/3 & \leq \|O_x \ket{\psi} - O_{x \oplus e_i} \ket{\psi}\|^2 \tag*{by Claim~\ref{claim: final states far}}\\
    & = \| O_x \ket{\phi_{1i}} + O_x \ket{\phi_{2i}} - O_{x \oplus e_i} \ket{\phi_{1i}} - O_{x \oplus e_i} \ket{\phi_{2i}} \|^2\\
    & = \| O_x \ket{\phi_{1i}} - O_{x \oplus e_i} \ket{\phi_{1i}} \|^2 \tag*{since $O_{x}$ and $O_{x \oplus e_i}$ have the same action on $\ket{\phi_{2i}}$}\\
    & \leq (\|O_x \ket{\phi_{1i}}\| + \|O_{x \oplus e_i} \ket{\phi_{1i}}\|)^2 \tag*{by the triangle inequality}\\
    & \leq \left(2\cdot\sqrt{\sum_{(a_1, \dots, a_T, b, w) \in J_i}\abs{\alpha_{a_1, \dots, a_T, b, w}}^2}\right)^2.
\end{align*}
Summing over all $i \in I$, this implies
\begin{align*}
    \Omega(\abs{I}) \leq \sum_{i \in [n]}\sum_{(a_1, \dots, a_T,b,w) \in J_i} \abs{\alpha_{a_1, \dots, a_T, b, w}}^2 \leq T \sum_{a_1, \dots, a_T,b,w} \abs{\alpha_{a_1, \dots, a_T, b, w}}^2 \leq T,
\end{align*}
where the second inequality follows because each $\abs{\alpha_{a_1, \dots, a_T, b, w}}^2$ can appear in the summation for at most $T$ values of $i$, and the last inequality follows since the sum of squares of the moduli of the amplitudes must equal 1. The theorem now follows.
\end{proof}

\begin{proof}[Proof of Theorem~\ref{thm: quantum NAANDT OMB}]
    Clearly $\QNAANDT(\OMB_n) \geq \QNAANDT(\OMB'_n)$. Claim~\ref{claim: omb' reduction} implies that $\QNAANDT(\OMB'_n) = \Qdt(\OMB'_n)$. Recall that the domain of $\OMB'_n$ consists only of the inputs $\cbra{x \in \zone^n : x = 0^i1^{n - i}~\text{for some}~i \in [n]}$. By definition, $\OMB'_n(0^i1^{n - i}) \neq \OMB'_n(0^{i-1}1^{n - i + 1})$ for all $i \in [n]$. Thus Theorem~\ref{thm: quantum dt lower bound} is applicable with $I = [n]$ and $f = \OMB'_n$. Combining the above, we have $\QNAANDT(\OMB_n) \geq \QNAANDT(\OMB'_n) = \Qdt(\OMB'_n) = \Omega(n)$.
\end{proof}
\begin{proof}[Proof of Theorem~\ref{thm: intro omb separation}]
It follows from Claim~\ref{claim: doneway omb} and Theorem~\ref{thm: quantum NAANDT OMB}.
\end{proof}

\subsection{Symmetric Functions}

In this section we show that symmetric functions $f$ admit efficient \textnaandt s in terms of the deterministic (even two-way) communication complexity of $f \circ \AND$.
We require the following bounds on the M\"obius sparsity of symmetric functions, due to Buhrman and de Wolf~\cite{BW01}. For a non-constant symmetric function $f : \zone^n \to \zone$, define the following measure which captures the smallest Hamming weight inputs before which $f$ is not a constant: $\sw(f) := \min\bra{k : f ~\textnormal{is a constant on all}~x~\textnormal{such that}~|x| < n-k}$.
\begin{claim}[{\cite[Lemma 5]{BW01}}]\label{claim: symmetric sparsity}
Let $n$ be sufficiently large,
let $f : \zone^n \to \zone$ be a symmetric Boolean function, and let $k := \sw(f)$.
Then $\log \spar(f) \geq \frac{1}{2}\log\left(\sum_{i = n - k}^n \binom{n}{i}\right)$.
\end{claim}

Upper bounds on the \textnaandt~complexity of symmetric functions follow from known results in the non-adaptive group testing literature. To the best of our knowledge, the following upper bounds were first shown (formulated differently) by Dyachkov and Rykov~\cite{DR83}. Also see~\cite{CH08} and the references therein. 

\begin{theorem}\label{thm: upper bound naandt in terms of switch symmetric}
Let $f : \zone^n \to \zone$ be a symmetric Boolean function with $\sw(f) = k < n/2$. Then $\NAANDT(f) = O\left(\log^2 \binom{n}{k}\right)$.
\end{theorem}
We give a self-contained proof of Theorem~\ref{thm: upper bound naandt in terms of switch symmetric} for clarity and completeness. The proof is via the probabilistic method. We construct a random family of $O\left(\log^2 \binom{n}{k}\right)$ many ANDs and argue that with non-zero probability, their evaluations on any input determine the function's value.

We require the following intermediate claim.
\begin{claim}\label{claim: intermediate probabilistic claim}
Let $n$ be a positive integer, and let $1 \leq k < n/2$ be an integer. Then, there exists a collection $\X$ of $O\left(\log^2 \binom{n}{k}\right)$ many subsets of $[n]$ satisfying the following.
\begin{equation}\label{eqn: properties of X}
\forall i_1, \dots, i_{k+1} \in [n], j \in [k+1], \exists X \in \X ~\text{such that}~ i_j \in X, i_{\ell} \notin X~\text{for all}~\ell \neq j.
\end{equation}
\end{claim}

\begin{proof}
Consider a random set $X \subseteq [n]$ chosen as follows: For each index $i \in [n]$ independently, include $i$ in $X$ with probability $1/(2k)$. Pick $w$ many sets (where $w$ is a parameter that we fix later) independently using the above sampling process, giving the multiset of sets $\X = \cbra{X_1, \dots, X_w}$.

For fixed $i_1, \dots, i_{k+1} \in [n]$, $j \in [k+1]$ and $t \in [w]$,
\begin{equation}\label{eqn: single k+1tuple prob bound}
\Pr_{X_t}[i_{j} \in X_t ~\text{and}~i_{\ell} \notin X_t ~\text{for all}~{\ell} \neq j] = \frac{1}{2k} \cdot \rbra{1 - \frac{1}{2k}}^k \geq \frac{1}{2k\cdot e},
\end{equation}
where the last inequality uses the fact that $k \geq 1$ and the standard inequality that $1-x \geq e^{-2x}$ for all $0 \leq x \leq 1/2$.
Thus Equation~\eqref{eqn: single k+1tuple prob bound} implies that for fixed $i_1, \dots, i_{k+1} \in [n]$ and $j \in [k+1]$,
\begin{equation}
\Pr_{\X}[\nexists X \in \X : i_{j} \in X ~\text{and}~i_{\ell} \notin X ~\text{for all}~{\ell} \neq j] \leq \rbra{1 - \frac{1}{2k \cdot e}}^w \leq \exp(-w/(2ke)).
\end{equation}
By a union bound over these ``bad events'' for all $i_1, \dots, i_{k+1} \in [n]$ and $j \in [k+1]$, we conclude that
\begin{align}
& \Pr_{\X}[\forall i_1, \dots, i_{k+1} \in [n] ~\text{and}~ j \in [k+1],~\exists X \in \X : i_{j} \in X ~\text{and}~i_{\ell} \notin X ~\text{for all}~{\ell} \neq j]\nonumber\\ 
& \geq 1 - \binom{n}{k+1} \cdot (k+1) \cdot \exp(-w/(2ke)). \label{eqn: union bound prob RHS}
\end{align}
We want to choose $w$ such that this probability is greater than 0. Thus we require
\begin{align*}
1 & > \binom{n}{k+1} \cdot (k+1) \cdot \exp(-w/(2ke))\\
\iff \exp(w/(2ke)) & > (k+1) \cdot \binom{n}{k + 1}\\
\iff w & > 2ke\rbra{\log (k+1) + \log\binom{n}{k+1}}.
\end{align*}
Since $\binom{n}{j+1} \geq n > j+1$ for all $j \in \bra{1, 2, 3, \dots, n/2}$ and $n > 2$, and since $\log\binom{n}{j} \geq j\log(n/j) \geq j$ for all $j \in \bra{1, 2, \dots, n/2}$, it suffices to choose
\begin{equation}\label{eqn: w choice}
w \geq 2e\log\binom{n}{k}\rbra{2\log \binom{n}{k+1}}.
\end{equation}
By standard binomial inequalities we have $\log\binom{n}{k+1} \leq (k+1)\log(ne/(k+1))$, and $\log\binom{n}{k} > k\log(n/k)$. Next, since $k+1 \leq 2k$ for $k \geq 1$ and $ne/(k+1) < n^3/k^3$ for $k \in \bra{1, 2, \dots, n/2}$, Equation~\eqref{eqn: w choice} implies that it suffices to choose
\[
w \geq 2e\log\binom{n}{k}\rbra{12\log\binom{n}{k}}.
\]
For this choice of $w$, the RHS of Equation~\eqref{eqn: union bound prob RHS} is strictly positive. This proves the claim.
\end{proof}

\begin{proof}[Proof of Theorem~\ref{thm: upper bound naandt in terms of switch symmetric}]
Let $f$ be a symmetric function with $\sw(f) = k < n/2$, and let $\X$ be as in Claim~\ref{claim: intermediate probabilistic claim} with $|\X| = O\rbra{\log^2\binom{n}{k}}$.
We now show how $\X$ yields a NAADT for $f$. Without loss of generality assume that $f(x) = 0$ for all $|x| < n-k$ (if not, output 1 in place of 0 in the \textbf{Output} step of Algorithm~\ref{algo: thn-k naandt} below).

\begin{algorithm}[H]\label{algo: thn-k naandt}
\SetAlgoLined
\textbf{Input:} $x \in \zone^n$ \\
    \begin{enumerate}
        \item Let $\X$ be as obtained from Claim~\ref{claim: intermediate probabilistic claim}.
        \item Query $\cbra{\AND_X(x) : X \in \X}$ to obtain a string $P_x \in \zone^{|\X|}$.
    \end{enumerate}
\textbf{Output:} $f(y)$ if $P_x = P_y$ for some $y$ with $|y| \geq n-k$, and 0 otherwise.
\caption{$\NAANDT$ for $f$}
\end{algorithm}
We show below that the following holds: $P_x \neq P_y$ for all $x \neq y \in \zone^n$ such that $|y| \geq n-k$.
This would show correctness of the algorithm as follows:
\begin{itemize}
    \item If $P_x = P_y$ for some $|y| \geq n - k$, then $x$ must equal $y$ by the above. In this case we output the correct value since we have learned $x$.
    \item If $P_x \neq P_y$ for any $|y| \geq n - k$, then $|x| < n - k$. Since $f$ evaluates to $0$ on all such inputs, we output the correct value in this case. 
\end{itemize}
Let $x \neq y \in \zone^n$ be two strings such that $|y| \geq n- k$. Without loss of generality assume $|y| \geq |x|$ (else swap the roles of $x$ and $y$ above). Let $I_x, I_y \subseteq [n]$ denote the sets of indices where $x$ and $y$ take value 0, respectively. By assumption, $x \neq y$ and $|I_x| \geq |I_y|$. Thus there exists an index $i_x \in I_x \setminus I_y$.

Since $|I_y| \leq k$, by Claim~\ref{claim: intermediate probabilistic claim} there exists $X \in \X$ such that $i_x \in X$ and $X \cap I_y = \emptyset$. Thus, for this $X$ we have
\[
\AND_X(x) = 0, \qquad \AND_X(y) = 1.
\]
Hence $P_x \neq P_y$, which proves the correctness of the algorithm and yields the theorem.
\end{proof}

\begin{remark}
The proof above in fact yields a NAADT of cost $O\left(\log^2 \binom{n}{k}\right)$ for any function $f : \zone^n \to \zone$ for which $f$ is a constant on inputs of Hamming weight less than $n-k$ for some $k < n/2$ (in particular, $f$ need not be symmetric on inputs of larger Hamming weight).
\end{remark}

We are now ready to prove Theorem~\ref{thm: intro symmetric naandt communication}.

\begin{proof}[Proof of Theorem~\ref{thm: intro symmetric naandt communication}]
If $\sw(f) \geq n/2$, then Claim~\ref{claim: symmetric sparsity} implies that $\spar(f) = 2^{\Omega(n)}$. Equation~\eqref{eqn: log mobius sparsity communication lower bound} implies that $\Dcc(f \circ \AND) = \Omega(n)$. Thus, a trivial NAADT of cost $n$ witnesses $\NAANDT(f) = O(\Dcc(f \circ \AND))$ in this case.

Hence, we may assume $\sw(f) = k < n/2$.
We have
\[
 \NAANDT(f) = O\left(\log^2 \binom{n}{k}\right) = O(\log^2(\spar(f))) = O(\Dcc(f \circ \AND)^2),
\]
where the first equality follows from Theorem~\ref{thm: upper bound naandt in terms of switch symmetric}, the second from Claim~\ref{claim: symmetric sparsity}, and the third from Equation~\eqref{eqn: log mobius sparsity communication lower bound}.
\end{proof}

\section*{Acknowledgements}
Swagato Sanyal thanks Prahladh Harsha and Jaikumar Radhakrishnan for pointing out the reference~\cite{FT99}, and Srijita Kundu for pointing out the reference~\cite{Kla00}. N.S.M.~thanks Ronald de Wolf for useful discussions. We thank Justin Thaler for pointing out a bug from an earlier version of the paper, and Arnab Maiti for helpful comments. We thank anonymous reviewers for helpful comments and suggestions, including an improvement of the value $b$ from 4 to 2 in Theorem~\ref{thm:detlift}.



\bibliography{main}

\newcommand{\etalchar}[1]{$^{#1}$}
\begin{thebibliography}{GMWW17}

\bibitem[AGJ{\etalchar{+}}17]{AGJ+17}
Anurag Anshu, Dmitry Gavinsky, Rahul Jain, Srijita Kundu, Troy Lee, Priyanka
  Mukhopadhyay, Miklos Santha, and Swagato Sanyal.
\newblock A composition theorem for randomized query complexity.
\newblock In {\em 37th {IARCS} Annual Conference on Foundations of Software
  Technology and Theoretical Computer Science (FSTTCS)}, pages 10:1--10:13,
  2017.

\bibitem[AJS20]{AJS20}
Prashanth Amireddy, Sai Jayasurya, and Jayalal Sarma.
\newblock Power of decision trees with monotone queries.
\newblock In {\em Proceedings od the Computing and Combinatorics - 26th
  International Conference (COCOON)}, pages 287--298, 2020.

\bibitem[AK98]{AK98}
Rudolf Ahlswede and Levon~H Khachatrian.
\newblock The diametric theorem in hamming spaces—optimal anticodes.
\newblock {\em Advances in Applied mathematics}, 20(4):429--449, 1998.

\bibitem[BB20]{BB20}
Shalev Ben{-}David and Eric Blais.
\newblock A tight composition theorem for the randomized query complexity of
  partial functions: Extended abstract.
\newblock In {\em 61st {IEEE} Annual Symposium on Foundations of Computer
  Science, {FOCS}}, pages 240--246. {IEEE}, 2020.

\bibitem[BdW01]{BW01}
Harry Buhrman and Ronald de~Wolf.
\newblock Communication complexity lower bounds by polynomials.
\newblock In {\em Proceedings of the 16th Annual {IEEE} Conference on
  Computational Complexity (CCC)}, pages 120--130, 2001.

\bibitem[Bei93]{Bei93}
Richard Beigel.
\newblock The polynomial method in circuit complexity.
\newblock In {\em Proceedings of the Eighth Annual Structure in Complexity
  Theory Conference}, pages 82--95, 1993.

\bibitem[BK16]{BK16}
Shalev Ben{-}David and Robin Kothari.
\newblock Randomized query complexity of sabotaged and composed functions.
\newblock In {\em 43rd International Colloquium on Automata, Languages, and
  Programming (ICALP)}, pages 60:1--60:14, 2016.

\bibitem[BT15]{BT15}
Mark Bun and Justin Thaler.
\newblock Dual lower bounds for approximate degree and markov-bernstein
  inequalities.
\newblock {\em Inf. Comput.}, 243:2--25, 2015.
\newblock Earlier version in ICALP'13.

\bibitem[BV97]{BV97}
Ethan Bernstein and Umesh~V. Vazirani.
\newblock Quantum complexity theory.
\newblock {\em {SIAM} J. Comput.}, 26(5):1411--1473, 1997.

\bibitem[BVW07]{BVW07}
Harry Buhrman, Nikolai~K. Vereshchagin, and Ronald~{de} Wolf.
\newblock On computation and communication with small bias.
\newblock In {\em 22nd Annual {IEEE} Conference on Computational Complexity
  (CCC)}, pages 24--32, 2007.

\bibitem[CFK{\etalchar{+}}21]{CFKMP21}
Arkadev Chattopadhyay, Yuval Filmus, Sajin Koroth, Or~Meir, and Toniann
  Pitassi.
\newblock Query-to-communication lifting using low-discrepancy gadgets.
\newblock {\em {SIAM} J. Comput.}, 50(1):171--210, 2021.

\bibitem[CH08]{CH08}
Hong{-}Bin Chen and Frank~K. Hwang.
\newblock A survey on nonadaptive group testing algorithms through the angle of
  decoding.
\newblock {\em J. Comb. Optim.}, 15(1):49--59, 2008.

\bibitem[CKLM19]{CKLM19}
Arkadev Chattopadhyay, Michal Kouck{\'{y}}, Bruno Loff, and Sagnik
  Mukhopadhyay.
\newblock Simulation theorems via pseudo-random properties.
\newblock {\em Comput. Complex.}, 28(4):617--659, 2019.

\bibitem[DR83]{DR83}
Arkadii~Georgievich Dyachkov and Vladimir~Vasil'evich Rykov.
\newblock A survey of superimposed code theory.
\newblock {\em Problems of Control and Information Theory}, 12(4):1--13, 1983.

\bibitem[dRNV16]{dRNV16}
Susanna~F. de~Rezende, Jakob Nordstr{\"{o}}m, and Marc Vinyals.
\newblock How limited interaction hinders real communication (and what it means
  for proof and circuit complexity).
\newblock In {\em {IEEE} 57th Annual Symposium on Foundations of Computer
  Science (FOCS)}, pages 295--304. {IEEE} Computer Society, 2016.

\bibitem[FT99]{FT99}
Peter Frankl and Norihide Tokushige.
\newblock The {E}rdos-{K}o-{R}ado theorem for integer sequences.
\newblock {\em Comb.}, 19(1):55--63, 1999.

\bibitem[GGKS20]{GGKS20}
Ankit Garg, Mika G{\"{o}}{\"{o}}s, Pritish Kamath, and Dmitry Sokolov.
\newblock Monotone circuit lower bounds from resolution.
\newblock {\em Theory Comput.}, 16:1--30, 2020.

\bibitem[GJ16]{GJ16}
Mika G{\"{o}}{\"{o}}s and T.~S. Jayram.
\newblock A composition theorem for conical juntas.
\newblock In {\em 31st Conference on Computational Complexity (CCC)}, pages
  5:1--5:16, 2016.

\bibitem[GLM{\etalchar{+}}16]{GLMWZ16}
Mika G{\"{o}}{\"{o}}s, Shachar Lovett, Raghu Meka, Thomas Watson, and David
  Zuckerman.
\newblock Rectangles are nonnegative juntas.
\newblock {\em {SIAM} J. Comput.}, 45(5):1835--1869, 2016.

\bibitem[GLSS19]{GLSS19}
Dmitry Gavinsky, Troy Lee, Miklos Santha, and Swagato Sanyal.
\newblock A composition theorem for randomized query complexity via
  max-conflict complexity.
\newblock In {\em 46th International Colloquium on Automata, Languages, and
  Programming (ICALP)}, volume 132, pages 64:1--64:13, 2019.

\bibitem[GMWW17]{GMWW17}
Dmitry Gavinsky, Or~Meir, Omri Weinstein, and Avi Wigderson.
\newblock Toward better formula lower bounds: The composition of a function and
  a universal relation.
\newblock {\em {SIAM} J. Comput.}, 46(1):114--131, 2017.

\bibitem[GPW18]{GPW18}
Mika G{\"{o}}{\"{o}}s, Toniann Pitassi, and Thomas Watson.
\newblock Deterministic communication vs. partition number.
\newblock {\em {SIAM} J. Comput.}, 47(6):2435--2450, 2018.

\bibitem[GPW20]{GPW20}
Mika G{\"{o}}{\"{o}}s, Toniann Pitassi, and Thomas Watson.
\newblock Query-to-communication lifting for {BPP}.
\newblock {\em {SIAM} J. Comput.}, 49(4), 2020.
\newblock Earlier version in FOCS'17.

\bibitem[HHL18]{HHL18}
Hamed Hatami, Kaave Hosseini, and Shachar Lovett.
\newblock Structure of protocols for {XOR} functions.
\newblock {\em {SIAM} J. Comput.}, 47(1):208--217, 2018.
\newblock Earlier version in FOCS'16.

\bibitem[HLS07]{HLS07}
Peter H{\o}yer, Troy Lee, and Robert Spalek.
\newblock Negative weights make adversaries stronger.
\newblock In {\em Proceedings of the 39th Annual {ACM} Symposium on Theory of
  Computing (STOC)}, pages 526--535, 2007.

\bibitem[HLY19]{HLY19}
Kaave Hosseini, Shachar Lovett, and Grigory Yaroslavtsev.
\newblock Optimality of linear sketching under modular updates.
\newblock In Amir Shpilka, editor, {\em 34th Computational Complexity
  Conference, {CCC} 2019, July 18-20, 2019, New Brunswick, NJ, {USA}}, volume
  137 of {\em LIPIcs}, pages 13:1--13:17. Schloss Dagstuhl - Leibniz-Zentrum
  f{\"{u}}r Informatik, 2019.

\bibitem[Kla00]{Kla00}
Hartmut Klauck.
\newblock On quantum and probabilistic communication: Las vegas and one-way
  protocols.
\newblock In {\em Proceedings of the Thirty-Second Annual {ACM} Symposium on
  Theory of Computing (STOC)}, pages 644--651, 2000.

\bibitem[KLMY21]{KLMY20}
Alexander Knop, Shachar Lovett, Sam McGuire, and Weiqiang Yuan.
\newblock Log-rank and lifting for and-functions.
\newblock In {\em Proceedings of 53rd Annual {ACM} {SIGACT} Symposium on Theory
  of Computing}, pages 197--208. {ACM}, 2021.

\bibitem[KMSY18]{KMSY18}
Sampath Kannan, Elchanan Mossel, Swagato Sanyal, and Grigory Yaroslavtsev.
\newblock Linear sketching over f{\_}2.
\newblock In Rocco~A. Servedio, editor, {\em 33rd Computational Complexity
  Conference, {CCC} 2018, June 22-24, 2018, San Diego, CA, {USA}}, volume 102
  of {\em LIPIcs}, pages 8:1--8:37. Schloss Dagstuhl - Leibniz-Zentrum
  f{\"{u}}r Informatik, 2018.

\bibitem[KN97]{KN97}
Eyal Kushilevitz and Noam Nisan.
\newblock {\em Communication complexity}.
\newblock Cambridge University Press, 1997.

\bibitem[Kun17]{Kun17}
Srijita Kundu.
\newblock One-way quantum communication complexity with inner product gadget.
\newblock {\em Electron. Colloquium Comput. Complex.}, 24:152, 2017.
\newblock Comment \#1 to TR17-152.

\bibitem[LM19]{LM19}
Bruno Loff and Sagnik Mukhopadhyay.
\newblock Lifting theorems for equality.
\newblock In {\em 36th International Symposium on Theoretical Aspects of
  Computer Science (STACS)}, volume 126 of {\em LIPIcs}, pages 50:1--50:19.
  Schloss Dagstuhl - Leibniz-Zentrum f{\"{u}}r Informatik, 2019.

\bibitem[Lov16]{Lov16}
Shachar Lovett.
\newblock Communication is bounded by root of rank.
\newblock {\em J. {ACM}}, 63(1):1:1--1:9, 2016.

\bibitem[LS88]{LS88}
L{\'{a}}szl{\'{o}} Lov{\'{a}}sz and Michael~E. Saks.
\newblock Lattices, m{\"{o}}bius functions and communication complexity.
\newblock In {\em 29th Annual Symposium on Foundations of Computer Science
  (FOCS)}, pages 81--90, 1988.

\bibitem[MO09]{MO09}
Ashley Montanaro and Tobias Osborne.
\newblock On the communication complexity of {XOR} functions.
\newblock {\em CoRR}, abs/0909.3392, 2009.

\bibitem[Mon10]{Mon10}
Ashley Montanaro.
\newblock Nonadaptive quantum query complexity.
\newblock {\em Inf. Process. Lett.}, 110(24):1110--1113, 2010.

\bibitem[Mon14]{Mon14}
Ashley Montanaro.
\newblock A composition theorem for decision tree complexity.
\newblock {\em Chicago J. Theor. Comput. Sci.}, 2014, 2014.

\bibitem[NC00]{NC00}
Michael~A. Nielsen and Isaac~L. Chuang.
\newblock {\em Quantum Computation and Quantum Information}.
\newblock Cambridge University Press, 2000.

\bibitem[Rei11]{Rei11}
Ben Reichardt.
\newblock Reflections for quantum query algorithms.
\newblock In {\em Proceedings of the Twenty-Second Annual {ACM-SIAM} Symposium
  on Discrete Algorithms (SODA)}, pages 560--569, 2011.

\bibitem[RM99]{RM99}
Ran Raz and Pierre McKenzie.
\newblock Separation of the monotone {NC} hierarchy.
\newblock {\em Comb.}, 19(3):403--435, 1999.

\bibitem[San17]{San17}
Swagato Sanyal.
\newblock One-way communication and non-adaptive decision tree.
\newblock {\em Electron. Colloquium Comput. Complex.}, 24:152, 2017.

\bibitem[San19]{San19}
Swagato Sanyal.
\newblock {F}ourier sparsity and dimension.
\newblock {\em Theory of Computing}, 15(11):1--13, 2019.

\bibitem[She12]{She12}
Alexander~A. Sherstov.
\newblock Making polynomials robust to noise.
\newblock In {\em Proceedings of the 44th Symposium on Theory of Computing
  Conference (STOC)}, pages 747--758, 2012.

\bibitem[She13]{She13}
Alexander~A. Sherstov.
\newblock Approximating the {AND-OR} tree.
\newblock {\em Theory Comput.}, 9:653--663, 2013.

\bibitem[Tal13]{Tal13}
Avishay Tal.
\newblock Properties and applications of boolean function composition.
\newblock In {\em Innovations in Theoretical Computer Science (ITCS)}, pages
  441--454, 2013.

\bibitem[VC71]{VC}
VN~Vapnik and A~Ya Chervonenkis.
\newblock On the uniform convergence of relative frequencies of events to their
  probabilities.
\newblock {\em Theory of Probability \& Its Applications}, 16(2):264--280,
  1971.

\bibitem[Wol02]{Wol02}
Ronald~{de} Wolf.
\newblock Quantum communication and complexity.
\newblock {\em Theoretical Computer Science}, 287(1):337--353, 2002.

\bibitem[Wu21]{Wu21}
Hsin-Lung Wu.
\newblock On the communication complexity of {AND} functions.
\newblock {\em IEEE Transactions on Information Theory}, 2021.

\bibitem[WYY17]{WYY17}
Xiaodi Wu, Penghui Yao, and Henry~S. Yuen.
\newblock Raz-mckenzie simulation with the inner product gadget.
\newblock {\em Electron. Colloquium Comput. Complex.}, 24:10, 2017.

\bibitem[Yao79]{Yao79}
Andrew~Chi{-}Chih Yao.
\newblock Some complexity questions related to distributive computing
  (preliminary report).
\newblock In {\em Proceedings of the 11h Annual {ACM} Symposium on Theory of
  Computing (STOC)}, pages 209--213, 1979.

\bibitem[{\v{Z}}LG21]{ZLG21}
Julius {\v{Z}}ilinskas, Algirdas Lan{\v{c}}inskas, and Mario~R Guarracino.
\newblock Pooled testing with replication as a mass testing strategy for the
  covid-19 pandemics.
\newblock {\em Scientific Reports}, 11(1):1--7, 2021.

\end{thebibliography}

\appendix

\section{Addressing Function}\label{app: addressing}

Recall that Theorem~\ref{thm: san19} shows that for all Boolean functions $f$, we have $\NAPDT(f) = O(\sqrt{r} \log r)$, where $r$ denotes the Fourier sparsity of $f$. Moreover a quadratic separation is witnessed by the Addressing function. We showed in Claim~\ref{claim: max naandt OMB} that such an upper bound does not hold in the M\"obius-world, and $\NAANDT(\OMB_n) \in \bra{\spar(\OMB_n), \spar(\OMB_n) - 1}$.
In this section we show that the Addressing function already witnesses that a separation similar to that in Theorem~\ref{thm: san19} cannot hold in the M\"obius-world (even on allowing randomization in the decision trees). While the bound we obtain here is weaker than that in Claim~\ref{claim: max naandt OMB}, it is interesting that the Addressing function, which witnesses a quadratic separation in the Fourier-world, no longer does so in the M\"obius-world.

The following is our main claim of this section.

\begin{claim}\label{claim: addressing naandt}
Let $n \geq 2$ be a positive integer that is a power of 2. Then,
\[
\QNAANDT(\ADDR_n) = \Theta(\spar(\ADDR_n)^{\frac{1}{\log 3}}).
\]
\end{claim}

To show this, we show a lower bound on the \textqnaandt~complexity of the Addressing function, and we compute its sparsity exactly.
\begin{claim}\label{claim: addr qnaandt large}
For an integer $n \geq 2$ that is a power of 2,
\[
\QNAANDT(\ADDR_n) = \Theta(n).
\]
\end{claim}

\begin{claim}\label{claim: addr sparsity}
For an integer $n \geq 2$ that is a power of 2,
\[
\spar(\ADDR_n) = n^{\log 3}.
\]
\end{claim}
Claim~\ref{claim: addressing naandt} follows immediately from the two above claims. We now prove these claims.

\begin{proof}[Proof of Claim~\ref{claim: addr qnaandt large}]
Let $\mathcal{T}$ be a quantum non-adaptive AND query algorithm of cost $c$ that computes $\ADDR_n$ to error at most $\eps = 1/3$. An input to $\ADDR_n$ consists of $\log n$ addressing bits and $n$ target bits. Let $I_1$ be the set of indices of the addressing bits and $I_2$ be the set of indices of the target bits.
The algorithm works with a state space $\ket{S_1, \dots, S_c}\ket{b}\ket{w}$, where each $S_j \subseteq I_1 \cup I_2$, $b \in \zone^c$ and the last register captures the workspace. The algorithm does the following on input $z \in \zone^{\log n + n}$:
\begin{itemize}
    \item It starts in an input-independent state, say $\sum_{S_1, \dots, S_c, b, w}\alpha_{S_1, \dots, S_c, b, w}\ket{S_1, \dots, S_c}\ket{b}\ket{w}$,
    \item applies $c$ non-adaptive queries, that act on the basis states as follows:
    \[
    O_z : \ket{S_1, \dots, S_c}\ket{b_1, \dots, b_c}\ket{w} \mapsto \ket{S_1, \dots, S_c}\ket{b_1 \oplus \AND_{S_1}(z), \dots, b_c \oplus \AND_{S_c}(z)}\ket{w}.
    \]
    \item It then applies a two-outcome projective measurement $\cbra{\Pi, I - \Pi}$ on the resulting state $O_z \ket{\psi}$ and outputs a value depending on the measurement outcome.
\end{itemize}

Fix any $i \in \zone^{\log n}$ and consider the inputs $z_i^0 := (i,0^n)$ and $z_i^1 := (i,s_i)$ where $s_i$ is the bitstring of length $n$ that has a $1$ at index $\bin(i)$ and $0$ everywhere else. Clearly $\ADDR_n(z_i^0) = 0$ and $\ADDR_n(z_i^1)=1$. We also define a set
\[
B_i := \cbra{S \subseteq I_1 \cup I_2 \mid S \cap I_2 = \cbra{\bin(i)}}.
\]
It represents the set of AND queries that query only the $\bin(i)$'th bit of the target variables. Note that for any $S \not\in B_i$, $\AND_S(z_i^0) = \AND_S(z_i^1)$. This is because any $\AND$ differentiating the two inputs must include the only bit in which they differ (i.e., $y_{\bin(i)}$), and if it includes any other target variables we know it must output $0$. It is easy to see that $B_i \cap B_j = \emptyset$ if $i \neq j$. Define
\[
J_i := \cbra{(S_1,\dots,S_c,b,w) \mid B_i \cap \cbra{S_1,\dots,S_c} \neq \emptyset} 
\]
to be (the indices of) those basis states which are mapped to different states on applying the oracles $O_{z_i^0}$ and $O_{z_i^1}$. Write the initial state of our algorithm as $\ket{\psi} = \ket{\phi_{1i}} + \ket{\phi_{2i}}$, where $\ket{\phi_{1i}} = \sum_{s \in J_i} \alpha_s \ket{s}$ and $\ket{\phi_{2i}} \perp \ket{\phi_{1i}}$. Thus,
\begin{align*}
    0.114 \approx 2 - 4\sqrt{2}/3 & \leq \|O_{z_i^0} \ket{\psi} - O_{z_i^1} \ket{\psi}\|^2 \tag*{by Claim~\ref{claim: final states far}}\\
    & = \| O_{z_i^0} \ket{\phi_{1i}} + O_{z_i^0} \ket{\phi_{2i}} - O_{z_i^1} \ket{\phi_{1i}} - O_{z_i^1} \ket{\phi_{2i}} \|^2\\
    & = \| O_{z_i^0} \ket{\phi_{1i}} - O_{z_i^1} \ket{\phi_{1i}} \|^2 \tag*{since $O_{z_i^0}$ and $O_{z_i^1}$ have the same action on $\ket{\phi_{2i}}$}\\
    & \leq (\|O_{z_i^0} \ket{\phi_{1i}}\| + \|O_{z_i^1} \ket{\phi_{1i}}\|)^2\tag*{by the triangle inequality}\\
    & \leq \left(2\cdot\sqrt{\sum_{(S_1, \dots, S_c, b, w) \in J_i}\abs{\alpha_{S_1, \dots, S_c, b, w}}^2}\right)^2.
\end{align*}
Since $B_k \cap B_\ell = \emptyset$ for all $k \neq \ell$, each set $S \subseteq I_1 \cup I_2$ can appear in at most one of the sets $\cbra{B_\ell}_{\ell \in \zone^{\log n}}$. Summing over all $i \in \zone^{\log n}$, this implies
\begin{align*}
    \Omega(n) \leq \sum_{i \in \zone^{\log n}}\sum_{(S_1, \dots, S_c, b, w) \in J_i} \abs{\alpha_{S_1, \dots, S_c,b,w}}^2 \leq c \sum_{S_1, \dots, S_c,b,w} \abs{\alpha_{S_1, \dots, S_c,b,w}}^2 \leq c.
\end{align*}
The second inequality above follows because of the following: since each $S_j$ can appear in at most one $B_k$, each $\abs{\alpha_{S_1, \dots, S_c, b, w}}^2$ would be included in the summation for at most $c$ values of $i$. The last inequality follows since the sum of squares of the moduli of the amplitudes must equal 1.
This completes the proof.
\end{proof}

\begin{proof}[Proof of Claim~\ref{claim: addr sparsity}]
Let $\Ind(E)$ denote the indicator function of $E$, that is, $\Ind(E) = 1$ if $E$ is true, and 0 otherwise. From the expression in Definition~\ref{defn: addressing} we have
\begin{align}\label{eqn: addr expansion}
\ADDR_n(x, y) = \sum_{b \in \zone^{\log n}} y_b \Ind[x = b] = \sum_{b \in \zone^{\log n}} y_b \prod_{i \in [\log n] : b_i = 0}(1 - x_i)\prod_{i \in [\log n] : b_i = 1}x_i.
\end{align}
The monomials arising from each summand are disjoint, since monomials containing $y_b$ only appear in the summand corresponding to $b$. For all $b \in \zone^{\log n}$, the number of monomials in $\prod_{i \in [\log n] : b_i = 0}(1 - x_i)\prod_{i \in [\log n] : b_i = 1}x_i$ equals $2^{\log n - |b|}$, where $|b|$ equals the Hamming weight (number of 1s) of $b$. From the expansion in Equation~\eqref{eqn: addr expansion}, we obtain
\[
\spar(\ADDR_n) = \sum_{b \in \zone^{\log n}}2^{\log n - |b|} = \sum_{j = 0}^{\log n}\binom{\log n}{j}2^j = 3^{\log n} = n^{\log 3}.
\]
\end{proof}

\begin{claim}\label{claim: addr patcomp lower bound}
For an integer $n$ that is a power of 2,
\[
\mpat(\ADDR_n) \geq 2^n.
\]
\end{claim}
\begin{proof}
    Define the set of monomials $M := \cbra{y_b \prod_{i \in [\log n]}x_i}_{b \in [n]}$. Observe from Equation~\eqref{eqn: addr expansion} that all these monomials are in the M\"obius support of $\ADDR_n$. Consider the set of inputs $I := \cbra{(1^{\log n},s)}_{s \in \zone^n}$. We claim that each input in $I$ generates a unique pattern, even when restricting to monomials in $M$. Indeed, the monomial $y_b \prod_{i \in [\log n]}x_i$ evaluated on $(1^{\log n},s)$ is equal to $s_b$. Hence the pattern generated by $(1^{\log n},s)$ on the monomials in $M$ is the string $s$, and so there are at least $2^n$ distinct patterns.
\end{proof}

We conclude the following.
\begin{corollary}\label{cor: doneway vs rank addressing}
For an integer $n$ that is a power of 2 and for $f = \ADDR_n$,
\[
\Doneway(f \circ \AND) \geq \rank(M_{f \circ \AND})^{\log_32}.
\]
\end{corollary}
\begin{proof}
\begin{align*}
    \Doneway(f \circ \AND) & = \lc\log(\mpat(f))\rc \tag*{by Claim~\ref{claim: doneway equals log pattern}}\\
    & \geq n \tag*{by Claim~\ref{claim: addr patcomp lower bound}}\\
    & = \spar(f)^{\log_32} \tag*{by Claim~\ref{claim: addr sparsity}}\\
    & = \rank(M_{f \circ \AND})^{\log_32} \tag*{by Equation~\eqref{eqn: mobius sparsity rank}}.
\end{align*}
\end{proof}

\section{Derivation of Theorem~\ref{thm: packing}}
\label{erdos}
Recall Theorem~\ref{thm: packing}, restated below.
\begin{theorem}[Restatement of Theorem~\ref{thm: packing}]\label{thm: appendix packing}
Let $q \geq 3$ and $1 \leq d \leq n/3$. Let $\A \subseteq [q]^n$ be such that for all $x^{(1)}=(x^{(1)}_1, \ldots, x^{(1)}_n)$, $x^{(2)}=(x^{(2)}_1, \ldots, x^{(2)}_n) \in \A$, $|\{i \in [n] \mid x^{(1)}_i = x^{(2)}_i\}| \geq d$. Then, $|\A| < q^{n-\frac{d}{10}}$.
\end{theorem}

Let $q$ be as in the statement of the theorem. For $z \in \{0,1\}^n$ and $\sS \subseteq [n]$, let $z_\sS$ denote the restriction of $z$ to the indices in $\sS$. Let $|z|$ denote the Hamming weight of $z$, which is $|\{i \in [n] \mid z_i=1\}|$. 

For an arbitrary alphabet $L$, a set $\cH \subseteq L^n$ is called $d$-intersecting if for each $x=(x_i)_{i \in [n]}, x'=(x'_i)_{i \in [n]} \in \cH$, $|\{i \in [n] \mid x_i=x'_i\}| \geq d$. Let $\agr (d,q,n)$ denote the size of a largest $d$-intersecting set in $[q]^n$. Ahlswede and Khachatrian~\cite{AK98} and Frankl and Tokushige~\cite{FT99} independently determined $\agr(d,q,n)$ in their works.

For an integer $r \leq (n-d)/2$, let $\A_r$ be the following $d$-intersecting family in $\{0,1\}^n$.
\[\A_r:=\{z \in \{0,1\}^n \mid |z_{\{1,\ldots,d+2r\}}| \geq d+r\}.\]

Now consider the following $d$-intersecting family $\B_r$ in $[q]^n$: A string $x \in [q]^n$ belongs to $\B_r$ iff there exists a string $z \in \A_r$ such that for each $i \in [n]$, $z_i=1 \Rightarrow x_i=1$. $\B_r$ is easily seen to be $d$-intersecting. Hence for each such $r$, $\agr(d,q,n) \geq |\B_r|$.

\cite{AK98} and \cite{FT99} showed that in fact there is a choice of $r$ for which $\agr(d,q,n) = |\B_r|$. In other words, there exists a choice of $r$ such that $\B_r$ is a largest $d$-intersecting family in $[q]^n$.

\begin{theorem}[Theorem 2 in \cite{FT99}]
\label{FT}
Let $q \geq 3, d \geq 1, r=\lfloor \frac{d-1}{q-2} \rfloor$ and $n \geq d+2r$. Then, $\agr(d,q,n)=|\B_r|$.
\end{theorem}
Note that since $n \geq 3d$, it is also true that $n \geq d + 2r$. Proving Theorem~\ref{thm: packing} now amounts to estimating $|\B_r|$. A string in $\B_r$ can be generated as follows.
\begin{itemize}
\item Choose a subset $T \subseteq [d+2r]$ of size $d+r$.
\item For each $i \in T$, set $x_i=1$.
\item For each $i \notin T$, set $x_i$ arbitrarily.
\end{itemize}
There are $\binom{d+2r}{d+r}$ choices of $T$. For each choice of $T$, there are $q^{n-d-r}$ ways of assigning variables with indices outside $T$. We thus have,
\begin{align}
|\B_r| &\leq \binom{d+2r}{d+r}\cdot q^{n-d-r} \nonumber \\
&\leq 2^{d+2r} \cdot q^{n-d-r} \nonumber \\
& = q^{n-d(1-\frac{1}{\log q})-r(1-\frac{2}{\log q})} \label{inter}
\end{align}
When $q \geq 4$, this gives us $|\B_r| \leq q^{n - \frac{d}{2}}$.
When $q = 3$, the value of $r$ is $d-1$ and so $$|\B_r| \leq q^{n-d(1-\log_3 2)-(d-1)(1-2\log_3 2)} = q^{n-2d+3d\log_3 2 + 1 - 2\log_3 2} \leq q^{n-\frac{d}{10}}.$$

We now give a self-contained and simple proof of the statement of Theorem~\ref{thm: packing} for the special case of $q > (en/d)^2$ (with a worse constant).

Let $\X \subseteq [q]^n$ be such that every $x,x' \in \X$ agree in at least $d$ locations. Observe that each pair $(x,x')$ can be uniquely specified by,
\begin{itemize}
\item A set $P_{x,x'} \subseteq [n]$ of indices of size $d$ such that $x_i=x'_i$ for each $i \in P_{x,x'}$.
\item A vector $\textbf{a}=(a_i)_{i \in P_{x,x'}} \in [q]^d$. $\textbf{a}$ represents $x_{P_{x,x'}}=x'_{P_{x,x'}}$.
\item Vectors $x_{\overline{P_{x,x'}}}$ and $x'_{\overline{P_{x,x'}}}$.
\end{itemize}
Thus the number of pairs $(x,x')$ is at most the number of such representations, which is upper bounded by $\binom{n}{d} \cdot q^d \cdot q^{2(n-d)} \leq (en/d)^d \cdot q^{2n-d} <q^{2n-\frac{d}{2}}$ (since $q > (\frac{en}{d})^2$). Thus $|\X|^2 < q^{2n-\frac{d}{2}} \Rightarrow |\X| < q^{n-\frac{d}{4}}$.

\end{document}